\documentclass[10pt,letterpaper]{article}
\usepackage{cite}
\usepackage{amsmath,amssymb,amsfonts}
\usepackage{algorithmic}
\usepackage{graphicx}
\usepackage{mathrsfs}
\usepackage{textcomp}
\def\BibTeX{{\rm B\kern-.02em{\sc i\kern-.025em b}\kern-.08em
    T\kern-.1667em\lower.7ex\hbox{E}\kern-.125emX}}
 \usepackage{epsfig}
 \usepackage[latin1]{inputenc}
\usepackage{graphicx}      
\usepackage{amsmath,amssymb}
\usepackage{pstricks,pst-node,pst-text,pst-3d,multido,pst-plot,pst-coil}
\usepackage[normalem]{ulem}
\usepackage{hyperref}

\usepackage{caption}
\usepackage{float}
\usepackage{subfloat}
\usepackage{psfrag}
\usepackage{multicol}

\newtheorem{remark}{Remark}
\newtheorem{proposition}{Proposition}

\newtheorem{assumption}{Assumption}
\newtheorem{lemma}{Lemma}

\newcommand{\col}{ \mbox{col} }

\def\calb{{\cal B}}

\def\calf{{\cal F}}
\def\cala{{\cal A}}

\def\calo{{\cal O}}

\def\lef[{\left[\begin{array}}
\def\rig]{\end{array}\right]}
\def\qed{\hfill$\Box \Box \Box$}

\def\rea{\mathbb{R}}

\def\rea{\mathbb{R}}
\def\bfe{{\bf e}}

\def\begequ{\begin{equation}}
\def\endequ{\end{equation}}
\def\lab{\label}
\def\begite{\begin{itemize}}
\def\endite{\end{itemize}}
\def\begarr{\begin{array}}
\def\endarr{\end{array}}
\def\begequarr{\begin{eqnarray}}
\def\endequarr{\end{eqnarray}}
\def\diag{\mbox{diag}}

\def\intnum{\mathbb{Z}}

\newcommand{\bfq}{\bf q}

\def\L2{{\cal L}_2}
\def\L2e{{\cal L}_{2e}}

\def\bul{\noindent $\bullet\; $}

\def\rea{\mathbb{R}}

\def\intnum{\mathbb{Z}}

\def\diag{\mbox{diag}}

\def\diag{\mbox{diag}}

\def\col{\mbox{col}}

\def\et{\varepsilon_t}
\def\diag{\mbox{diag}}

\def\begmat#1{\begin{bmatrix}#1\end{bmatrix}}
\def\begali#1{\begin{align}{#1}\end{align}}
\def\begalis#1{\begin{align*}{#1}\end{align*}}
\def\begsubequ{\begin{subequations}}
\def\endsubequ{\end{subequations}}
\def\begequarr{\begin{eqnarray}}
\def\endequarr{\end{eqnarray}}
\def\begequarrs{\begin{eqnarray*}}
\def\endequarrs{\end{eqnarray*}}
\def\begarr{\begin{array}}
\def\endarr{\end{array}}
\def\begequ{\begin{equation}}
\def\endequ{\end{equation}}
\def\lab{\label}
\def\begdes{\begin{description}}
\def\enddes{\end{description}}
\def\begenu{\begin{enumerate}}
\def\begite{\begin{itemize}}
\def\endite{\end{itemize}}
\def\endenu{\end{enumerate}}

\def\lef[{\left[\begin{array}}
\def\rig]{\end{array}\right]}
\def\qed{\hfill$\Box \Box \Box$}
\def\begcen{\begin{center}}
\def\endcen{\end{center}}
\def\begrem{\begin{remark}\rm}
\def\endrem{\end{remark}}
\def\begassums{\begin{assumption*}}
\def\endassums{\end{assumption*}}
\def\begassu{\begin{ass}}
\def\endassu{\end{ass}}
\def\beglem{\begin{lemma}}
\def\endlem{\end{lemma}}
\def\begcor{\begin{corollary}}
\def\endcor{\end{corollary}}
\def\begfac{\begin{fact}}
\def\endfac{\end{fact}}
\def\begass{\begin{assumption}}
\def\endass{\end{assumption}}

\def\begmat#1{\begin{bmatrix}#1\end{bmatrix}}
\def\begali#1{\begin{align}{#1}\end{align}}
\def\begalis#1{\begin{align*}{#1}\end{align*}}
\def\TAC{{\it IEEE Trans. Automatic Control}}

\def\TPE{{\it IEEE Trans. Power Electronics}}

\def\IJC{{\it International Journal of Control}}

\def\SCL{{\it Systems and Control Letters}}
\def\AUT{{\it Automatica}}

\def\TCS{{\it IEEE Trans. on Circuits and Systems}}

\usepackage{color}

\def\begsubequ{\begin{subequations}}
	\def\endsubequ{\end{subequations}}
\def\begpro{\begin{proposition}}
	\def\endpro{\end{proposition}}
\def\beglem{\begin{lemma}}
	\def\endlem{\end{lemma}}
\def\begass{\begin{assumption}}
	\def\endass{\end{assumption}}
\def\begproo{\begin{proof}}
	\def\endproo{\end{proof}}
\def\begfac{\begin{fact}}
\def\endfac{\end{fact}}
\newtheorem{corollary}{Corollary}
\usepackage{subcaption}
\begin{document}

\title{An Algebraic State Observer  for a Class of Physical Systems}
\author{Alexey Bobtsov\thanks{A. Bobtsov and A. Pyrkin are with the Control Systems and Robotics Department, ITMO, St. Petesburg, Russia (e-mail: bobtsov\{a.pyrkin\}@itmo.ru).}
 , Jose Guadalupe Romero$^\dagger$,  Romeo Ortega\thanks{J.G. Romero and R. Ortega  are with  the Department of Electrical and Electronic Engineering,  ITAM, R\'io Hondo 1, Mexico City, 01080, Mexico (e-mail: \{jose.romerovelaquez\}\{romeo.ortega\}@itam.mx). }
 and Anton Pyrkin$^*$
}
\maketitle
\begin{abstract}
In this paper we present a radically new approach to design state observers for nonlinear systems, with particular emphasis on physical ones. Our objective is to obtain an {\em algebraic} relation between the unmeasurable part of the state and filtered versions of the systems inputs and outputs, which holds true {\em for all $t \geq 0$}. The latter qualifier should be contrasted with the usual {\em asymptotic} (or fixed/finite time) objective. The standing assumption for our design is the availability---or possibility of constructing, via coordinate change---state components with {\em measurable derivatives}. In the physical systems studied in the paper this condition is naturally satisfied. The next step in the design is the application of the {\em Swapping Lemma} to pull out from the dynamics the derivative of one of these signals. The design is completed replacing the latter by the measurable signals and arranging the remaining terms. The algebraic observer constitutes a refreshing major departure from classical asymptotic observer designs, even in the case of electrical motors and mechanical systems that have been exhaustively studied. Particularly notable is the fact that no {\em observability} or {\em excitation} condition is imposed for the construction of the algebraic observer. 
\end{abstract}


%%%%%%%%%%%%%%%%
\section{Introduction and Problem Formulation}
\lab{sec1}
%%%%%%%%%%%%%%%%
The problem of reconstruction of the state of a dynamical system out of the measurement of its inputs and outputs is one of the central problems of control engineering practice intensively investigated by the control community. Many monographs and research papers have been devoted to this problem with the successful identification of particular classes of nonlinear systems for which constructive solutions are available. This classes include systems in affine-in-the-state form and systems with triangular structure---the interested reader is referred to the recent book \cite{BERbook} and the tutorial paper \cite{BERANDAST} for a review of the recent developments on this topic.  One feature of all existing solutions to this problem is that the reconstruction of the state vector is achieved {\em asymptotically}, and in some recent designs---that rely on the practically undesirable injection of high-gain---in {\em finite/fixed time}, see {\em e.g.} \cite{EFIPOL} for a recent tutorial on this topic. Unfortunately, most of the physical systems of practical interest, including mechanical and electromechanical systems, do not fit into the two aforementioned systems structures. Moreover, due to inevitable presence of noise, they are very sensitive to the injection of high gain.
   
We consider in the paper nonlinear systems described in the standard $(f,g,h)$ form \cite{ISIbook}, that is:
\begequ
\lab{fgh}
\begin{cases}
\dot x =f(x)+g(x)u\\
y=h(x),
\end{cases}
\endequ
with $x(t) \in \rea^n$, $u(t) \in \rea^m$ and $y(t) \in \rea^p$. In this paper we take a radically different approach to the state observation problem---our objective being to obtain an {\em algebraic} relation between a part of the state that is {\em unmeasurable}, denoted $x_u(t) \in \rea^{n_u},\;n_u \leq n$, and {\em filtered} versions of the systems inputs and outputs, with this identity holding true {\em for all $t \geq 0$}.  That is, we want to find a {\em mapping} $\calo: \rea^{2m+2p} \to \rea^{n_u}$ such that the following identity holds true\footnote{Since the filters time constants cand be chosen arbitrarily small, we neglect its transient in the identity \eqref{staide}.}  
\begequ
\lab{staide}
x_u(t)=\calo[u(t),y(t),\calf[u](t),\calf[y](t)],\;\forall t \geq 0,
\endequ
where $\calf(p) \in \rea(p)$ is a linear time-invariant (LTI) filter. Notice that, in contrast to the standard definition of an observer, we want the identity \eqref{staide} to remain valid {\em for all $t \geq 0$}, and not {\em asymptotically} as $t \to \infty$. On the other hand, we also require that the identity \eqref{staide} holds true {\em for all} initial conditions and input signals---attaching in this way the qualifier {\em global} to the result.

The main standing assumption for the solution of this problem is the existence of some of the state variables, whose {\em time derivative is measurable}. Taking into account, that if such states are not available, we can create them via coordinate change that, as expected, involves the solution of two partial differential equations. Indeed, it is possible to replace the assumption of existence of elements of the vector $x$ whose derivative is measurable, by the following.

\begass
\lab{ass1}
There exists mappings 
$$
\varpi: \rea^n \to \rea^{n_m},\;w_f :\rea^p \to \rea^{n_m},\;w_g:\rea^p \to \rea^{n_m},
$$
such that the following equations are satisfied
\begalis{
\nabla^\top \varpi(x) f(x)&=w_f(h(x))\\
\nabla^\top \varpi(x) g(x)&=w_g(h(x)).
}
\endass

It is clear that if {\em Assumption 1} holds we have that the signal $z_m:=\varpi(x)$ satisfies
$$
\dot z_m = w_f(y)+w_g(y)u,
$$
whose right hand side is {\em measurable}.

It should be noted that the scenario where some of the state variables have measurable derivatives is common in physical systems. For instance, systems with magnetic components satisfy Faraday's law $\dot \phi=v$, with $\phi$ the {\em unmeasurable} flux and $v$ the {\em measurable} voltage. The same is true for electrical systems where the derivative of the charge, which is an {\em unmeasurable} quantity, is the current that is {\em measurable}.\footnote{See the recent paper \cite{ORTetal_scl25} where this property is exploited to design ``classical" adaptive state observers for systems with quadratic nonlinearities.}
 
 Once these signals with known time derivatives have been identified (or constructed) we appeal to the powerful tool of the {\em Swapping Lemma}  \cite[Lemma 3.6.5]{SASBODbook}. This lemma, which played a central law in the solution of the model reference adaptive control (MRAC) problem in the 1980s, was reported by Morse in the groundbreaking paper \cite{MOR}. Using this lemma it is possible to---via LTI filtering---extract from the product of two signals a term depending on the {\em derivative} of one of them. In the scenario of MRAC this signal is the parameter estimation error, whose derivative was known to be square integrable, a feature that turned out to be instrumental to provide the first complete global stability proof of the MRAC problem in  \cite{MOR}. In our algebraic observer design scenario, this signal is then replaced by the measurable ones and rearranged in such a way as to create the algebraic relation $\calo$ required by the design. Given the importance of the Swapping Lemma in our observer design, we refer to it in the sequel as Algebraic Swapping Lemma Observer (ASLO).  A particularly notable feature of ASLO is the fact that no {\em observability} or {\em excitation} condition is imposed for its construction---which constitutes a major departure from {\em all existing} observer designs.
 
 It is clear that, due to the use of {\em filtered signals}, a transient time is implicit in all the derivations---which, as indicated in Footnote 1 above, is neglected in \eqref{staide}. This transient can be modified selecting the filters time constants---however, given the complicated combinations of the filtered signals in the ASLO, this is done in a rather indirect way. To provide a clearer handle on the assignment of a convergence rate to the observers, besides the ASLO we present an asymptotically convergent variation of it that incorporates a tuning gain that directly affects the convergence time.  This new asymptotic observer---named in the sequel {\em asymptotic ASLO} (A-ASLO) directly incorporates the information of the ASLO, hence it also constitutes a radically new observer design.
 
 The format of presentation we have adopted in the paper is the following: first, we start with a simple {\em motivational example} in Section \ref{sec2}, that---given the novelty of the design procedure---is worked out in full detail.  The performance of the ASLO and the A-ASLO are compared in simulations with the standard asymptotically convergent Luenberger observer. Then, in Section \ref{sec3} we illustrate the ASLO desing procedure with and in Section \ref{sec4} with {\em mechanical systems}. It is worth pointing out that for all these applications, that have been exhaustively studied in the literature, the reported ASLO gives radically new solutions---with no assumptions whatsoever on observability imposed for the design. Finally, in Section \ref{sec5} we present some concluding remarks. The key Swapping Lemma is recalled in Appendix \ref{appa} and the proofs of the main propositions are given in Appendices \ref{appb} to \ref{appe}.\\

%\noindent {\bf Caveat Emptor} The present paper is a shortened version of a full one reported in {\tt (arXiv xxx)}, where other physical examples---including the generalized electric machine---are presented.
% 
\noindent {\bf Notation} For $b \in \rea^n$ we denote the Euclidean norm as $|b|$. $\rea_+$ and $\intnum_+$ denote the positive real and integer numbers, respectively. For $q \in \intnum_+$ we define the set $\bar q:=\{1,2,\dots,q\}$. ${\bf I}_n$ is the $n \times n$ identity matrix and ${\bf 0}_{n \times s}$ is an $n \times s$ matrix of zeros, while $\bfe_j$ is the $j$-th element of the Euclidean basis. The action of an LTI filter $\calf(p) \in \rea(p)$ on a signal $w(t)$ is denoted as $\calf[w](t)$, where $p^n[w]:={d^n w(t)\over dt^n}$. To avoid cluttering, when clear from the context, the time argument is omitted. Given a function $H:  \rea^n \to \rea$, we define the differential operator $\nabla H(x):=\left(\frac{\displaystyle \partial H}{\displaystyle \partial x}\right)^\top $.
%%%%%%%%%%%
%%%%%%%%%%%%%%%
%
%%%%%%%%%%%%%%%%
\section{Motivational Example}
\lab{sec2}
%%%%%%%%%%%%%%%%
%
In this section we illustrate the procedure that we will follow to design the ASLO with a simple academic example. \vspace{-0.3cm}
\subsection{Problem formulation}
\lab{subsec21}
%%%%%%%%%%%%%%%%
%
Consider the second order LTI system 
\begequ
\lab{simsys}
	\dot x_1=x_2,\;\dot x_2=u, \;y=x_1.
\endequ
The objective is to design, using the measurements of $u$ and $y$, an ASLO for the state $x_2$. That is, we want to find a mapping $\calo: \rea^2 \to \rea$ such that the following identity holds true
\begequ
\lab{x2ide}
x_2(t)=\calo(u(t),y(t),\calf[u](t),\calf[y](t)),\;\forall t \geq 0,
\endequ
where $\calf(p) \in \rea(p)$. Notice that, in contrast to the standard definition of an observer, we want the identity \eqref{x2ide} to hold true {\em for all $t \geq 0$} and not {\em asymptotically} as $t \to \infty$. The property required above is also different from {\em convergence in finite (or fixed) time} where we assume the existence (or fix) a time $t_c>0$ such that  \eqref{x2ide} holds for all $t \geq t_c$. As will become clear below, and in the examples given in the sequel, the construction of our ASLO is radically different from the standard construction of existing observers.
\subsection{Design of the ASLO}
\lab{subsec22}
%%%%%%%%%%%%%%%%
% 
\begpro
\lab{pro1}
Consider the system \eqref{simsys}. The ASLO of the state $x_2$ is given by\vspace{-0.3cm}
\begequ
\lab{aslox2}
x_2(t)=p \calf[y](t)+{1 \over \lambda}\calf[u](t),\;\forall t \geq 0,
\endequ
where we defined the LTI filter
\begequ
\lab{fil}
\calf(p)={\lambda \over p + \lambda},
\endequ
with $\lambda>0$.
\endpro

{\bf Proof 1}
Using the obvious identity $\dot x_2=p[x_2]+\lambda x_2-\lambda x_2,$ which, manipulating the factor $(p+\lambda)$, we rearrange as
\begali{
\lab{calfx2}
\calf[x_2] & =x_2-{1 \over \lambda}\calf[\dot x_2]=x_2-{1 \over \lambda}\calf[u]
}
where we used \eqref{simsys} to get the second identity. Invoking again the system dynamics \eqref{simsys} we can write the left hand side of \eqref{calfx2} as\vspace{-0.3cm}
\begali{
\calf[x_2] &= \calf[\dot x_1]= p\calf[x_1] = p \calf[y].
\lab{ydel}
}
Replacing this identity in \eqref{calfx2} and rearranging terms we obtain our ASLO \eqref{aslox2}.

\subsection{A new asymptotic observer}
\lab{subsec23}
%%%%%%%%%%%%%%%%
%
Setting large values for the constant $\lambda$ we can reduce the time of convergence of the filter $\calf(p)$.\footnote{It should be underscored that we have chosen the basic first order filter \eqref{fil} to simplify the derivations. However, as the SL is applicable to {\em any LTI system}, we can select any other LTI system---but there is no clear advantage for doing that.} However, its effect on the convergence rate of the ASLO is done in a rather indirect way. To provide a clearer handle on the assignment of the observers convergence rate, we present below an A-ASLO, that is, an {\em asymptotically convergent} variation of ASLO that explicitly incorporates a tuning gain that directly affects the convergence time.

\begin{corollary}
\lab{cor1}
Consider the system \eqref{simsys} and the ASLO given in \eqref{aslox2}. Define the asymptotic ASLO
\begequ
\lab{asyobsx2}
\dot {\hat x}_2=u-\gamma\Big[\hat x_2 -\Big(p \calf[y]+{1 \over \lambda}\calf[u]\Big)\Big],
\endequ
where $\gamma>0$ is a tuning gain. Then, the observation error $\tilde x_2:=\hat x_2-x_2$ satisfies $\dot {\tilde x}_2 = -\gamma {\tilde x}_2.$
\end{corollary}
\subsection{Discussion}
\lab{subsec24}
%%%%%%%%%%%%%%%%
%
We make the following observations regarding the derivations above.

\noindent {\bf O1} The claim that the identity \eqref{aslox2} {\em holds true} for all $t \geq 0$ should be understood taking into account that the action of the filters {\em is not} instantaneous. This fact is clearly seen if you consider in the example \eqref{simsys} the case of $u(t) \equiv 0,\forall t\geq 0$. In this case, which clearly corresponds to the case where $x_2(t)=x_2(0),\forall t\geq 0$, some simple calculations prove that the right hand side of \eqref{aslox2} yields instead $(1+e^{-\lambda t})x_2(0)$.   

\noindent  {\bf O2} It is interesting to compare the ASLO observer \eqref{aslox2} with the A-ASLO  \eqref{asyobsx2} and the classical reduced-order Luenberger observer \cite{RUGbook}. Towards this end, we obtain the  state-space realization of the three of them as:

\noindent {\bf ASLO \eqref{aslox2}}
 \begali{
\nonumber
      \begmat{{\dot  w}_1 \\ {\dot  w}_2} & =\begmat{-\lambda & 0 \\ 0 & -\lambda }\begmat{w_1 \\ {w}_2}+\begmat{\lambda & 0 \\ 0 & -\lambda^2}\begmat{u \\ y}\\
	 x_2 & =\begmat{\frac{1}{\lambda} & 1}\begmat{w_1 \\ {w}_2} + \lambda y.
\lab{aslorea}
}  

\noindent {\bf A-ASLO \eqref{asyobsx2}} 
\begequ
\lab{staspaasyobsx2}
\begmat{\dot w_1 \\ \dot w_2\\ \dot {\hat x}_2} =\begmat{-\lambda & 0 & 0\\ 0 & -\lambda & 0 \\ {\gamma \over \lambda} & \gamma & -\gamma}\begmat{ w_1 \\ w_2 \\ \hat x_2}+\begmat{\lambda & 0 \\ 0 & -\lambda^2 \\ 1 & \gamma \lambda }\begmat{u \\ y}.
\endequ

\noindent {\bf Reduced Order Luenberger Observer} \cite[Theorem 15.7]{RUGbook}
\begali{
\nonumber
\dot x_c & =-\gamma_L x_c+u-\gamma_L^2 y\\
\lab{lueobs}
\hat x_2 &=x_c+ \gamma_L y,
}  
with $\gamma_L>0$, whose error dynamics is also of the form $\dot {\tilde x}_2=-\gamma_L \tilde x_2$ as for the A-ASLO.

Comparing \eqref{aslorea} with \eqref{staspaasyobsx2} we see that---besides the presence of a new gain $\gamma$---the only difference is the addition of the state $\hat x_2$ which brings along a $u$-dependent term. Hence, it is expected that the latter will be more sensible to input noise, a fact that is verified in the simulations presented below. On the other hand, there is  {\em absolutely no resemblance} of neither one of these two observers with the Luenberger one rendering very difficult the comparison between them.  

\noindent  {\bf O3} It is possible to extend the result of Proposition \ref{pro1} for the case of {\em $n$-th order} systems. To avoid cluttering the notation, we illustrate the procedure with the case of a fourth order integrator,\footnote{With respect to the example of Proposition \ref{pro1}, we take for simplicity $\lambda=1$ in the filter $\calf(p)$ \eqref{fil}.} which can be easily extended for the $n$-th order case. That is, we consider:
$$
\dot x_1=x_2,\;\dot x_2=x_3,\;\dot x_3=x_4,\;\dot x_4=u,\;y=x_1,
$$
and our objective is to design an ASLO for $(x_2,x_3,x_4)$. From the equations above, the following identity is easily verified
$$ 
(p+1)^3x_4=p^3y+u+(p+1)u+(p+1)^2u,
$$
that, multiplying by ${1 \over (p+1)^3}$, can be rewritten as
\begequ
\lab{x4} 
x_4=w^y_3+w^u_3+w^u_2+w^u_2,
\endequ
where we defined the filtered signals
\begalis{
w^y_i &:={p^i \over (p+1)^i}y,\;i=1,2,3,\;w^u_i :={1 \over (p+1)^i}u,\;i=1,2,3.
}
Similarly, it is easy to verify that
\begali{
x_2&={p \over p+1}y+x_3-{1 \over p+1}x_4\;=w^y_1+x_3-x_4+w^u_1.
\lab{x2}
}
Proceeding in the same manner with $x_3$ we get
\begali{
\lab{x3}
x_3&=w^y_2+2x_4-w^u_2.
}
Replacing \eqref{x4} in \eqref{x3}, and this in turn in \eqref{x2}, completes the design of the ASLO.

\noindent  {\bf O4} It is interesting to note that the ASLO is {\em robust  vis-\`a-vis constant disturbances} at the input and output. Let us illustrate this fact with the second order example \eqref{simsys}. First, assume the {\em output} is perturbed by a constant disturbance $\delta$, that is, $y=x_1+\delta$, the state estimates of the ASLO are {\em insensitive} to the presence of this unknown disturbance. This fact is clear if in \eqref{ydel} we evaluate $p\calf[x_1-\delta]$ and see that---as $p[\delta]=0$---this is still equal to $p \calf[y]$. It is easy to prove that this robustness property holds true for the $n$-th order case system considered above. Second, suppose now that it is the input signal that is perturbed by an additive constant, that is $u \leftarrow u+\delta$. Replacing this in \eqref{calfx2} results in 
$$
x_2(t)=p\calf[y]+{1 \over \lambda}\calf[u+\delta]=p\calf[y]+{1 \over \lambda}\calf[u]+\et,
$$ 
where $\et \to 0$ exponentially fast---recovering asymptotically the true value of $x_2$. From \eqref{lueobs} it is clear that such a disturbance induces an unstable behavior in the Luenberger observer. 
\subsection{Simulation results}
\lab{subsec25}
%%%%%%%%%%%%%%%%
% 
In this subsection we evaluate the performance of the ASLO \eqref{aslox2}, the A-ASLO \eqref{asyobsx2} and reduced order Luenberger observer \eqref{lueobs} for different filter and tuning gains, and the presence of input measurement noise as well as dynamic uncertainty.

The simulations were carried out using the input signal   $u=\cos(0.02t)$. The initial conditions for the second-order system were set as  $x(0)=\col(0,0.1)$. For the three observers all initial conditions were set to zero in all simulations. For the three observers, we present the observer error $\tilde x_2$. 

To evaluate the effect of the {\em tuning parameters} on the transient performance, different values of  $\lambda$   were considered for the ASLO. As expected, as  $\lambda$ increases, the observer error $\tilde x_2$ converges faster---without any overshoot---as shown in Fig. \ref{aslo}. For the same values of 
$\lambda$ the A-ASLO was simulated with  $\gamma=2$. The corresponding transient response is shown in Fig. \ref{nao}, where asymptotic convergence is observed, with improved convergence rate as $\lambda$ increases. As expected, with increased values of $\gamma$ the convergence time of A-ASLO decreases as shown in in Fig. \ref{diffg}. For the Luenberger observer, simulations were conducted with $\gamma_L=\lambda$. The transient responses for different values of  $\gamma_L$ are presented in Fig. \ref{luen}. Surprisingly, and very hard to explain, the observer error dynamics of the ASLO and the Luenberger observer are {\it almost identical}, a behavior that was also observed in the robustness tests. 

Simulations were also conducted by adding {\em white noise} to the input signal, and fixing $\lambda=3$, for ASLO, $\gamma=2$ for A-ASLO, and $\gamma_L=3$ for the Luenberger observer. Fig. \ref{ruiu} shows that both observer errors remain near zero, with a slightly better response seen from ASLO. The transient behavior of the Luenberger observer is not presented, but it exhibits performance very similar to that of the ASLO.

Finally, the observers are tested in the presence of {\em dynamic parasitic disturbances} in the input signal, which are described by
 $$
 \tau \dot u = -u +v,
 $$
where $v$ is the input signal used by the observers. The simulations were conducted for different values of $\tau$  with $u(0)=0$,  and the same observer gains: $\lambda=3$ for ASLO, $\gamma=2$ for the A-ASLO and $\gamma_L=3$ for the Luenberger observer. Fig. \ref{diftau} shows the transient behavior of the observer errors for the ASLO and A-ASLO designs under different values of $\tau$. Both observers achieve satisfactory performance for small values of $\tau$, with ASLO systematically ensuring a better behavior. Moreover, for slower parasitic dynamics, the error of A-ASLO exhibits a drift away from zero. Similarly to the case of input noise, the behavior of the Luenberger observer is very similar to the one of ASLO design, therefore it is omitted.

\begin{figure*}[htp!]
	\centering
	\subcaptionbox{\label{aslo} ASLO for different $\lambda$ }{\includegraphics[width=0.34\textwidth]{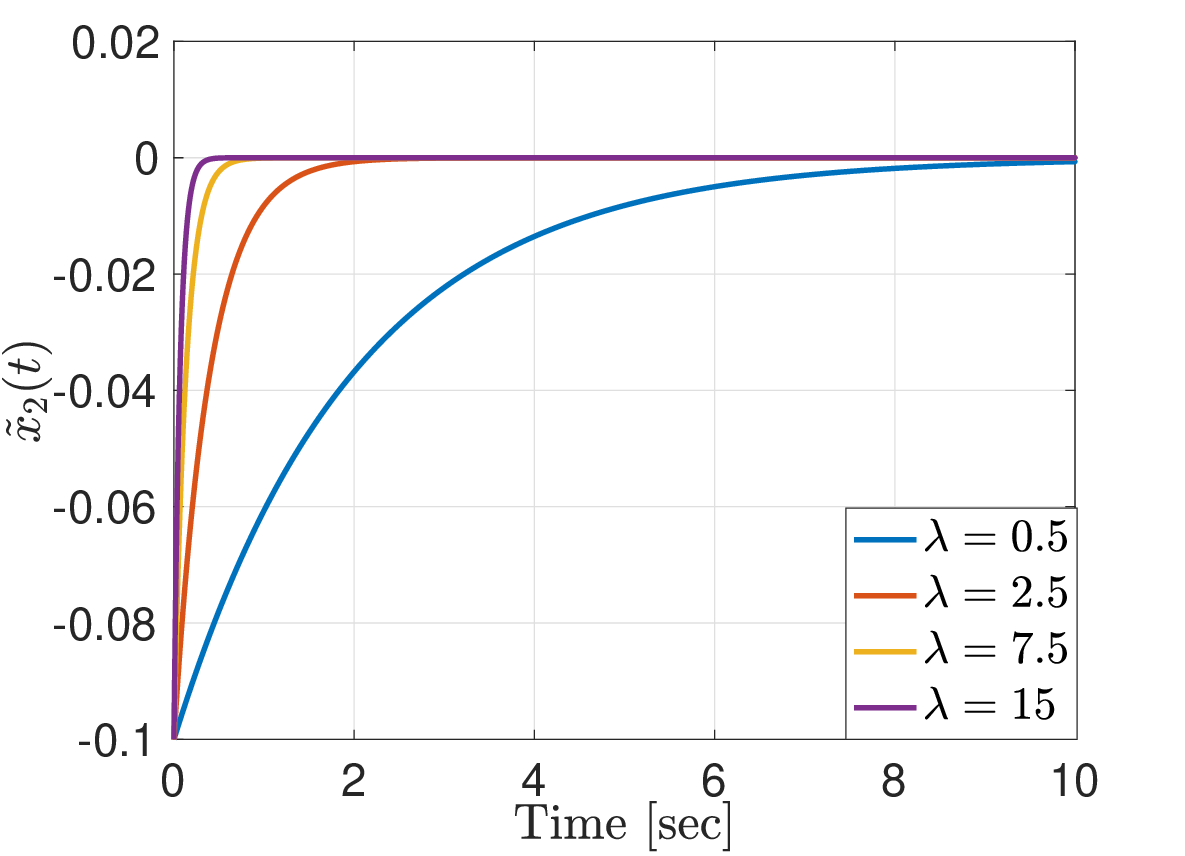}}
	\
	\subcaptionbox{\label{nao} A-ASLO with $\gamma=2$ and different $\lambda$}{\includegraphics[width=0.34\textwidth]{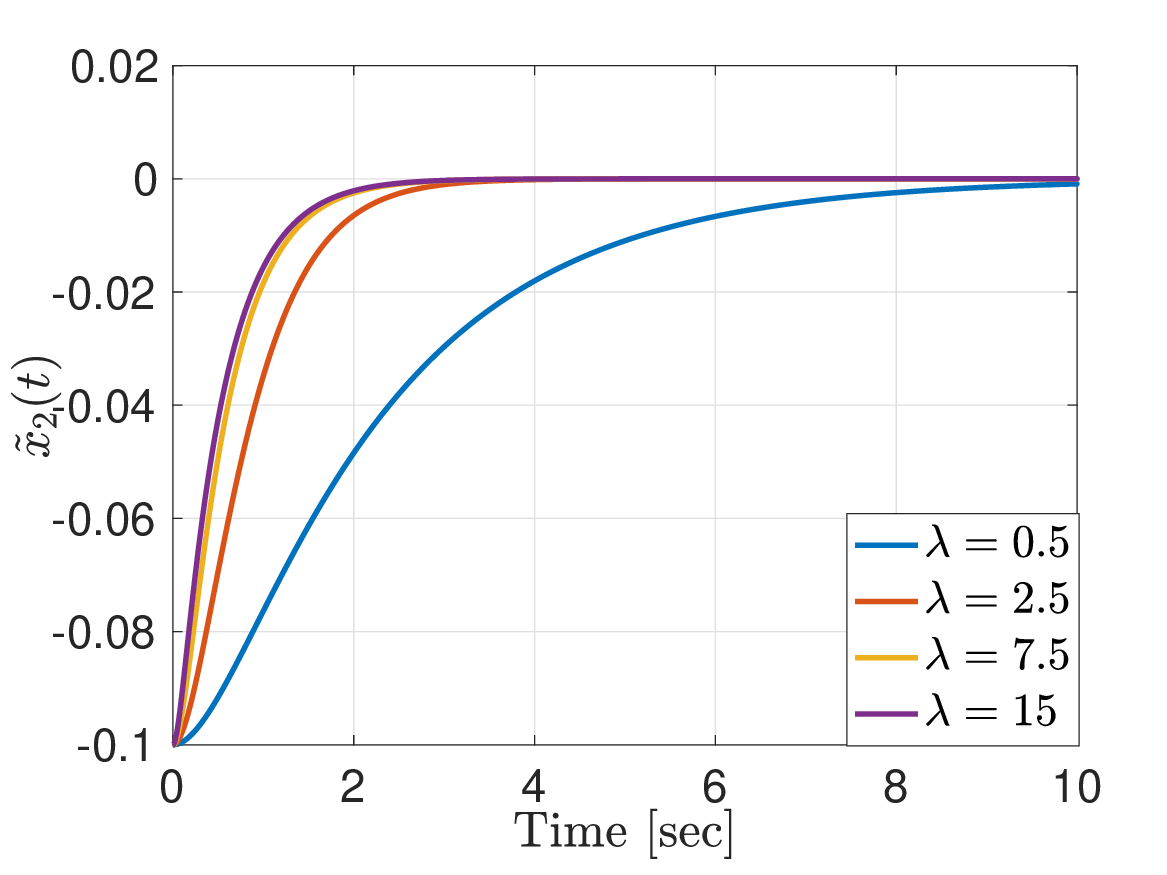}}
	\
		\subcaptionbox{\label{diffg} A-ASLO with $\lambda=2.5$ and different $\gamma$ }{\includegraphics[width=0.34\textwidth]{figs/aslo}}
	\
	\subcaptionbox{\label{luen} Luenberger observer with different $\gamma_L$}{\includegraphics[width=0.34\textwidth]{figs/asob}}
	\caption{\label{fig_1} Transient behavior of $\tilde x_2 $}
\end{figure*}

\begin{figure}[htp!]
    \centering
    \includegraphics[width=0.6\textwidth]{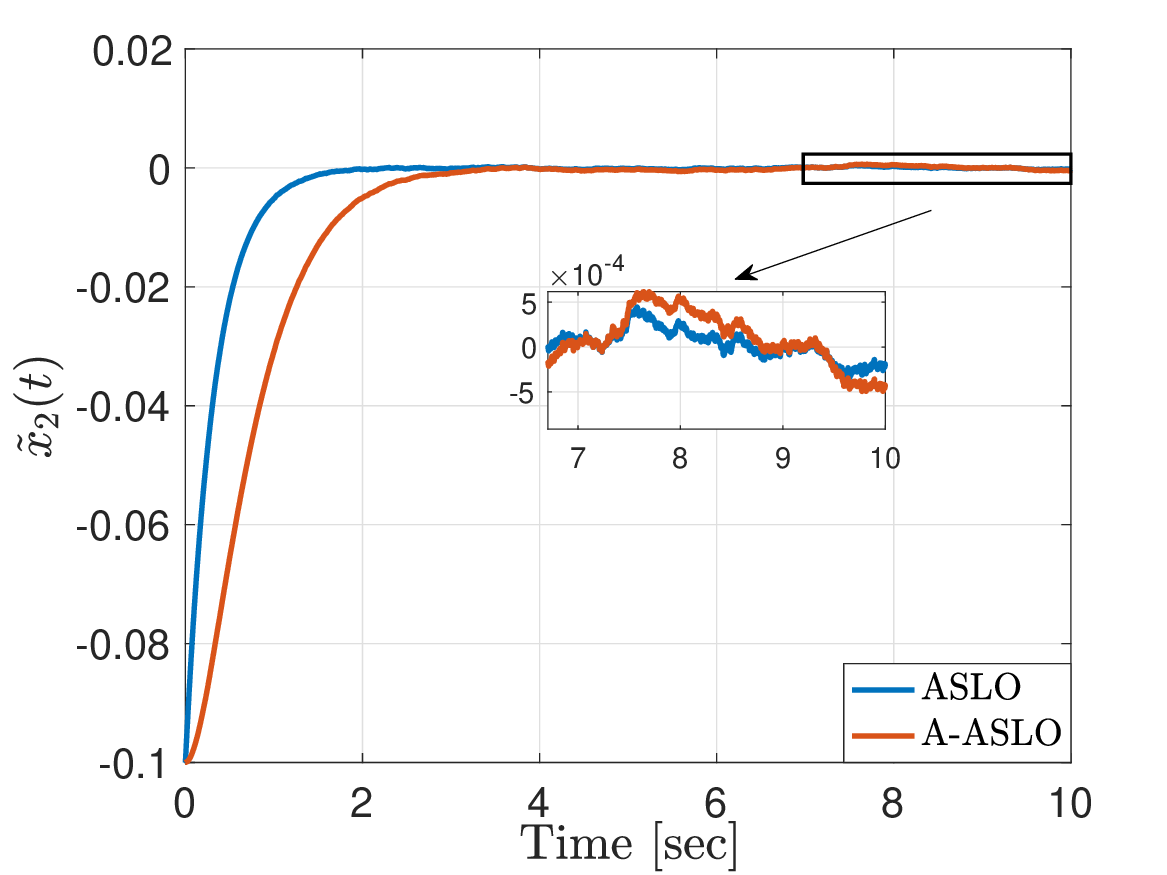}
    \caption{Transient behavior of $\tilde x_2 $ for ASLO and A-ASLO  with white noise added to $u$ } 
    \label{ruiu}
\end{figure}
 
\begin{figure}[htp!]
    \centering
    \includegraphics[width=0.8\textwidth]{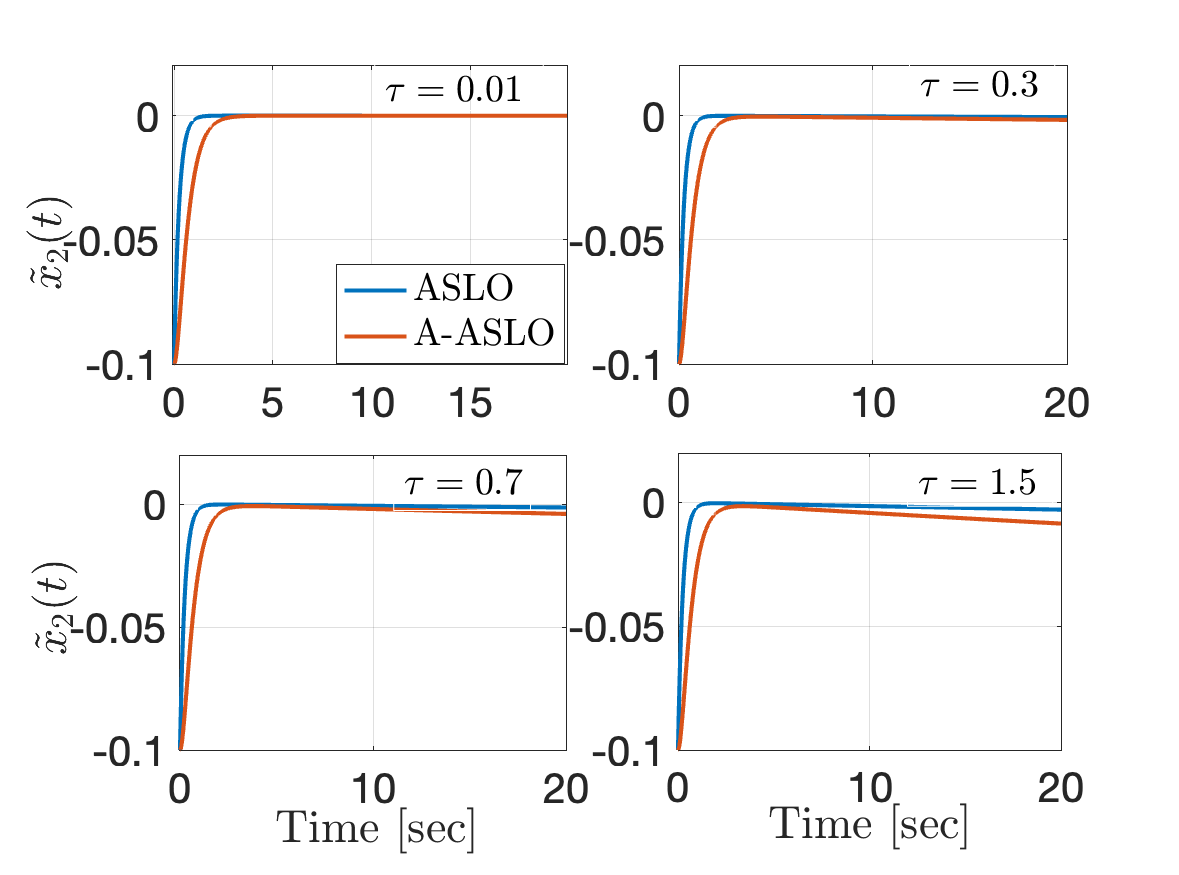}
    \caption{Transient behavior of $\tilde x_2 $ for ASLO and A-ASLO with input unmodelled dynamics with different time constant $\tau$ } 
    \label{diftau}
\end{figure}
%%%%%%%%%%%%%%%%%
%%%%%%%%%%%%%%%%
\section{Electric Motor Examples}
\lab{sec3}
%%%%%%%%%%%%%%%%
%
In this section we present the examples of two electric motors which admit the design of an ASLO as formulated in Section \ref{sec1}. Also, after reviewing the model of the generalized electrical motor, we discuss its potential extension to other classes of motors. We concentrate on the problem of observation of the {\em flux}. It is worth pointing out that, since it is possible to reconstruct the rotor position from the flux \cite[Eq. (9)]{LEEetal}, this is the crucial problem to be solved in the practically relevant scenario of {\em sensorless control} \cite{KRAbook,MARTOMVERbook,NAMbook}---that has been intensively studied by the control community. As will become clear below the mechanical dynamics, which is given here for completeness, plays no role on the design of the ASLO---this situation is similar to the design of other classes of flux observers, {\em e.g.}, \cite{BERbook,BOBPYRORT,CHOetal,LEEetal,MARTOMVERbook,PYRetal,YIetal}.

To streamline the presentation of the results, throughout the paper we introduce {\em measurable signals}---which are algebraic combinations of the inputs and outputs---that we identify with the letter $z_i$ and combinations of filtered measurable signals denoted $w_i$.  \vspace{-0.7cm}
%
%%%%%%%%%%%%%%%%
\subsection{Surface-mount permanent magnet synchronous motor (PMSM)}
\lab{subsec31}
%%%%%%%%%%%%%%%%
In this subsection we design an ASLO for the flux of the PMSM measuring only voltages and currents. The solution we will present here is {\em radically different} from all existing solutions reported in the literature. Assessing its performance opens a new, wide research avenue in this fundamental topic.

%%%
\subsubsection{Mathematical model}
\lab{subsec311}
%%%%%%%%%%%%%%%5
%
The classical, two-phase $\alpha\beta$ model of the non-salient PMSM is given by \cite[eqs. (1)-(3)]{BOBPYRORT} \vspace{-0.3cm}
\begin{eqnarray}
\nonumber 
\dot  \phi & = & -R y+ u\\
\nonumber
\dot  \theta &= & \omega \\
J\dot  \omega &=& -{R_m}  \omega + {n_p} (y_2  \phi_1 - y_1  \phi_2) -{\tau_L},
\label{sys}
\end{eqnarray}
$\phi(t) \in \rea^2$ is the stator flux vector, $\theta(t) \in [0,2\pi)$ is the rotor angle, $ \omega(t) \in \rea$ the rotor angular velocity, $y(t) \in \rea^2$ is the {stator current} and $u(t) \in \rea^2$ is the applied voltage.\footnote{Notice that the first two equations of the systems model \eqref{sys} are Faraday's Law.} The constant $R$ is the stator resistance, $J$ is the rotor moment of inertia, $R_m$ the friction coefficient, $n_p$ the number of pole pairs and $\tau_L(t) \in \rea$ is the load torque. In the practically relevant scenario of {\em sensorless control} \cite{CHOetal}, the only measurable quantity is the {\em stator current}, which is related to the systems state via \cite[eqs. (2) and (3)]{BOBPYRORT}: \vspace{-0.4cm}
\begequ
\lab{flucur}
L y=\phi-{\lambda_m} \begmat{\cos(n_p \theta) \\ \sin(n_p \theta)},
\endequ
with $L$ the stator inductance and $\lambda_m$ the flux of the permanent magnet.
%   
%%%%%%%%%%%%%%%%%
\subsubsection{Main Result}
\lab{subsec312}
%%%%%%%%%%%%%%
%
In this section we present the ASLO for the stator flux of the PMSM and give an A-ASLO. The proof of the proposition is given in Appendix \ref{appb}. 

\begpro 
\lab{pro2}
Consider the dynamics of the PMSM system given in \eqref{sys}. Define the following {\em measurable} signals \vspace{-0.15cm}
\begali{
z_1&:=- R y_1+ u_1,\;z_2:=- R y_2+ u_2 \nonumber\\
z_3&:=-\frac{\lambda_m^2}{2L} + \frac{L}{2}(y_1^2 +y_2^2).
\lab{z}
}
and introduce the {\em dynamic extension}
\begalis{
w_4&={-{p \over \lambda}\calf[y_1]}+{1 \over \lambda L}\calf[z_1],  \;w_5={-{p \over \lambda}\calf[y_2]}+{1 \over  \lambda L}\calf[z_2]\\
w_6&=  {1 \over \lambda }\calf[z_1\calf[y_1]]+{1 \over \lambda }\calf[z_2\calf[y_2]]+{1 \over \lambda L}\calf[z_1{1 \over \lambda}\calf[z_1]]+ \\
&\;\;{1 \over \lambda L}\calf[z_2{1 \over \lambda}\calf[z_2]]-{{p \over \lambda}\calf[z_3]} \\
w_7&=\calf[w_4],\;w_8=\calf[w_5],\\
w_9&=\calf[w_6]+{1 \over \lambda}\calf[z_1\calf[w_4]] +{1 \over \lambda}\calf[z_2\calf[w_5]],
}
with $\calf(p)$ the LTI filter defined in \eqref{fil}. Also, define the signal 
\begequ
\lab{del}
\Delta:=w_4w_8-w_5w_7, 
\endequ
which we assume is different to zero for all $t \geq 0$.\footnote{If there exists $t_c \in \rea_+$ such that the generic condition is violated, standard modifications---{\em e.g.} temporary freeze the estimation or add a jump \cite{SASBODbook}---can be incorporated to avoid the singularity.}
 
\noindent {\bf P1} The ASLO of $\phi$ is given by
\begin{equation}
\phi(t)={1 \over \Delta(t)} \begmat{w_6(t) w_8(t)-w_5(t)w_9(t)\\ w_4(t) w_9(t)-w_6(t)w_7(t) }.
\label{aslophipm}
\end{equation}

\noindent {\bf P2} Define the A-ASLO flux observer 
\begali{
\nonumber
	\dot {\hat \phi}_1 &=z_1-\gamma\hat \phi_1+ \gamma {1 \over \Delta}(w_6 w_8-w_5w_9)\\
\lab{asyobsphi}
	\dot  {\hat \phi}_2 &=z_2-\gamma \hat \phi_2+ \gamma {1 \over \Delta}(w_4 w_9-w_6w_7),
}
where $\gamma>0$ is a tuning gain. Then, the observation error $\tilde \phi:=\hat \phi-\phi$ satisfies 
\begin{equation}
\label{errequphi} 
\dot {\tilde \phi} = -\gamma {\tilde \phi}.
\end{equation}
\endpro
%%%%%%%%%%%%%%%%%%%%%%%%%%%%%%%%
\subsubsection{Simulation results}
\lab{subsec313}
%%%%%%%%%%%%%%%%%%%%%%%%%%%%%%
%
In this section we present simulation results of the flux ASLO and A-ASLO of Proposition \ref{pro2}. As indicated above this flux observation problem has attracted a lot of attention in the control community, therefore we compare the performance of three other  observers, namely:

\noindent {\bf FO1:} The observer reported in \cite{ABAO2Automatica}, which is given by the state equations\footnote{For $i,j \in \intnum_+,\;j>i$, we use the notation $\xi_{ij}:=\col(\xi_i,\xi_{i+1},\dots,\xi_j)$.}
\begin{align}
\dot \xi_{14} &=  \lef[{c} v \\ i \rig],\; \xi_{14}(0) \in \rea^4 \nonumber \\
\dot \xi_{5} & =   -\lambda \left( \xi_{5} +|\bfq|^2\right),\;\xi_{5}(0)\in \rea \nonumber \\
\dot \xi_{67} & =   -\lambda \left( \xi_{67} - 2 \bfq \right),\;\xi_{67}(0)\in \rea^2 \nonumber \\
\dot{\hat\eta} & =  \gamma \Omega (y - \Omega^\top \hat\eta),\;\hat \eta(0) \in \rea^2,
\label{flux_Automatica}
\end{align}
with the algebraic relations
\begin{align*}
\bfq & =  \xi_{12} - R \xi_{34} - Li,\;y   =   -\lambda  \left( |\bfq|^2 +  \xi_{5} \right),\\
\Omega  &=   \lambda  \left(2 \bfq -  \xi_{67} \right),\;\hat\phi = Li+\bfq+\hat\eta,
\end{align*}
where $\lambda>0$ is a filter parameter, $\gamma>0$ is the adaptation gain.

\noindent {\bf FO2:}  The observer reported in  \cite{ORTetal}, given by 
\begin{equation}
	\dot{\hat\phi}=v-Ri-\frac{\gamma}{2}{\partial h(\hat\phi,i)\over \partial {\hat\phi}}\max\{0,h(\hat\phi,i)\},
	\label{flux_Praly}
\end{equation}
where $h(\hat\phi,i):=|\hat\phi-Li|^2-\lambda_m^2$, with $\gamma>0$ the adaptation gain. It should be underscored that, in contrast with the observer proposed in this paper, the observer {\bf FO2} assumes that, besides the electrical parameters $R$ and $L$, the magnet flux constant $\lambda_m$ is also known.

\noindent {\bf FO3:}  A variation of the ASLO based on the estimation-based approach of \cite{PYRetal} given by\vspace{-0.2cm}
\begin{align}
	\dot {\hat \phi}_1 &=z_1+\gamma\Delta(w_6 w_8-w_5w_9-\Delta \hat \phi_1)\nonumber \\
	\dot  {\hat \phi}_2 &=z_2+\gamma\Delta(w_4 w_9-w_6w_7-\Delta \hat \phi_2),
	\label{asyobsphi2}
\end{align}
where $\gamma>0$ is tuning gain. Clearly, the observer satisfies $ \dot {\tilde \phi} = -\gamma \Delta^2 {\tilde \phi},$ which should be contrasted with the error dynamics of A-ASLO \eqref{errequphi}. Another interesting feature of this variation of A-ASLO is that the division by the determinant $\Delta$, that may cross through singularity, is {\em avoided}. 

In Table  \ref{tabla1}  we list the parameters of the PMSM used in the simulations, which were taken from \cite{data}. In all scenarios we study the behavior of the observers for the motor driven by a controller with the desired speed  $\omega_{ref}(t)=2+0.5\sin (t)$ [rad/s] and  load torque $\tau_L=3+0.5\sin(0.1 t)$ [Nm]. {\em For all} observers we selected the same tuning parameters $\lambda=5$ and $\gamma=5$. The regulation of the speed, as well as its measurement, was the industry standard PI control with a PLL for the reconstruction of the speed as reported in \cite[Section 2.B and 2.C]{LEEetal}.

Figs. \ref{d}-\ref{a_aslo2} show that all observer errors converge to zero as predicted by the theory. The convergence rate of all observers is significantly different, with the {\em fastest rate}---and smaller overshoot---achieved by the ASLO, followed closely by the ASLO-based observer {\bf FO3}. To improve the behavior of the observer {\bf FO2} it was necessary to take very large values of $\gamma$---of the order of $1000$.

\begin{table}[!hbt]
\begin{center}
\begin{tabular}{|l|c|}
\hline
Inductance $L(mH)$ & 40.03 \\
Resistance $R(\Omega)$ & 8.875 \\
Drive inertia $J(kgm^2)$ &  $60\times 10^{-6}$ \\
Pairs of poles $n_p (-)$ & 5 \\
Magnetic flux $\lambda_m(Wb)$ &  0.2086 \\
Maximum current $I(A)$ & 2.3 \\
\hline
\end{tabular}
\caption{Parameters of the motor BMP0701F used in the simulations}
\label{tabla1}
\end{center}
\end{table}
%%%%%%%%%%%%%%

\begin{figure*}[htp!]
	\centering
	\subcaptionbox{\label{d} Flux estimation error ASLO \eqref{aslophipm} }{\includegraphics[width=0.34\textwidth]{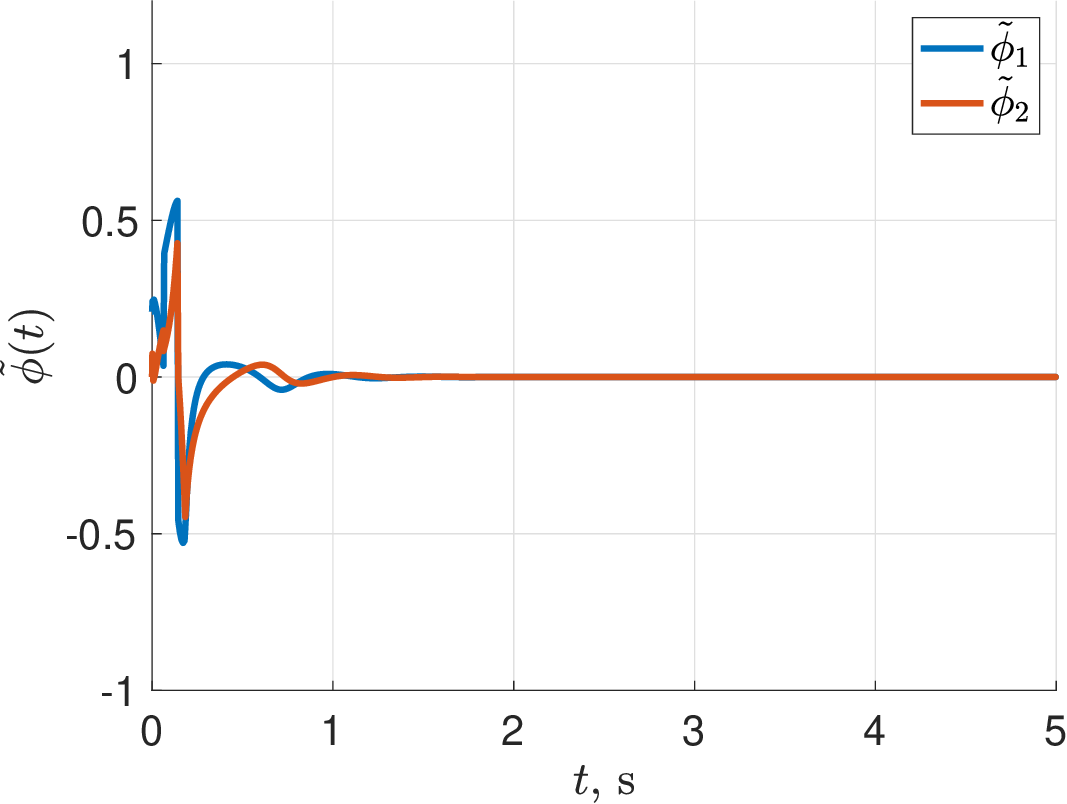}}
	\
	\subcaptionbox{\label{a_aslo} Flux estimation error A-ASLO \eqref{asyobsphi}}{\includegraphics[width=0.34\textwidth]{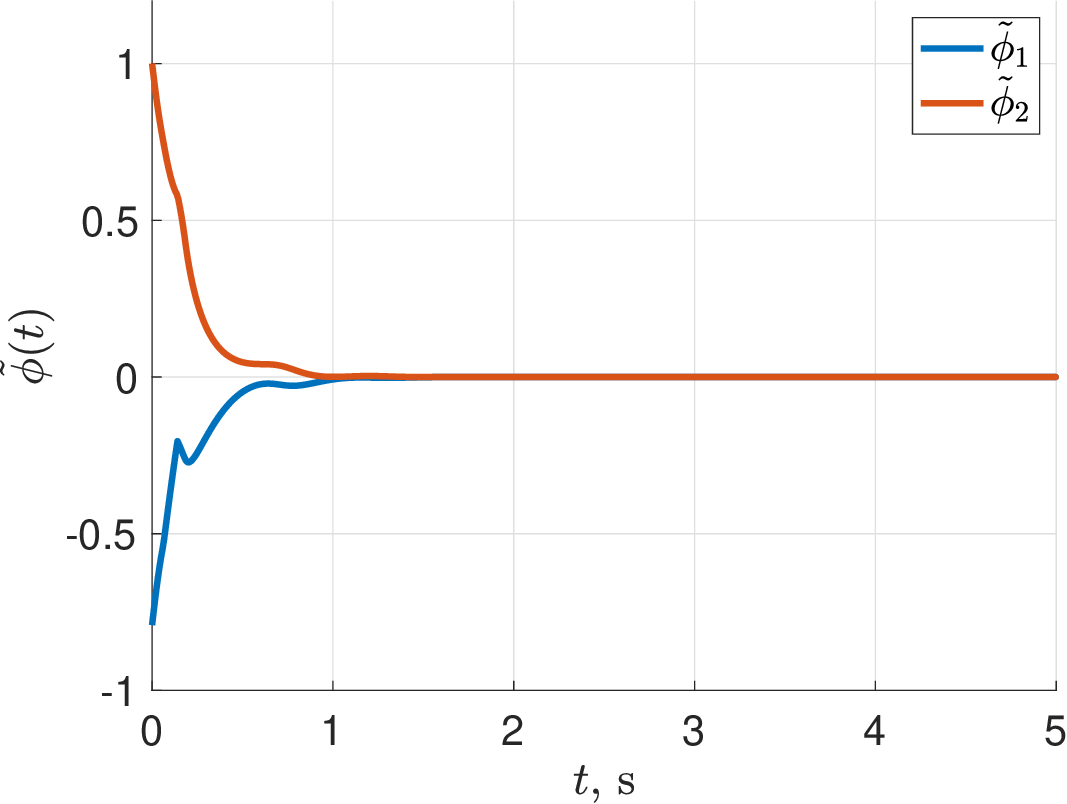}}
	\
	\subcaptionbox{\label{abao2} Flux estimation error for {\bf FO1} \eqref{flux_Automatica} }{\includegraphics[width=0.3\textwidth]{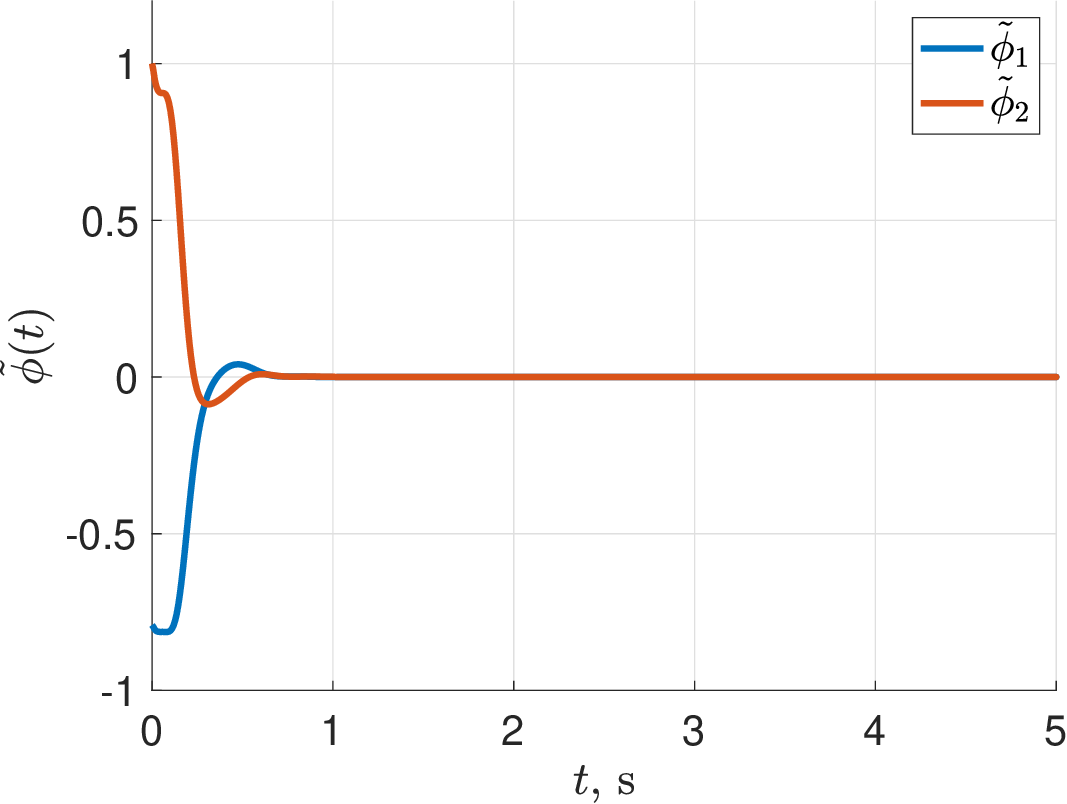}}
	\
	\subcaptionbox{\label{praly} Flux estimation error for {\bf FO2} \eqref{flux_Praly}}{\includegraphics[width=0.3\textwidth]{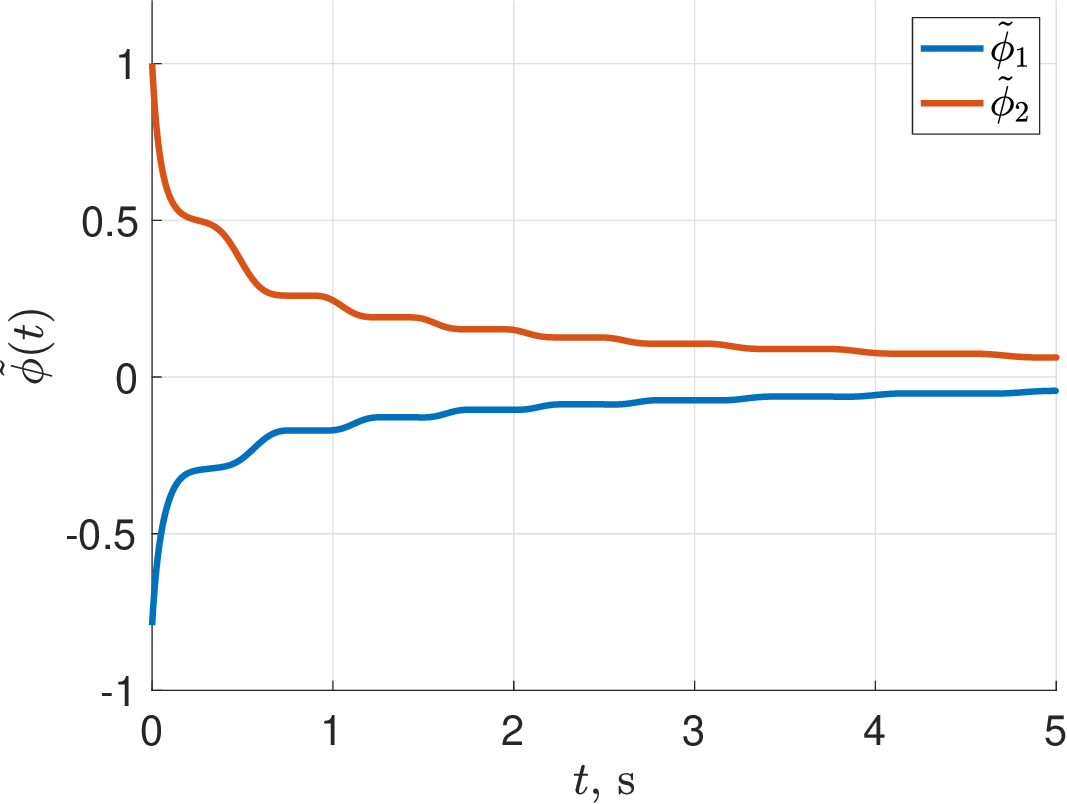}}
		\
	\subcaptionbox{\label{a_aslo2} Flux estimation error for {\bf FO3} \eqref{asyobsphi2}}{\includegraphics[width=0.3\textwidth]{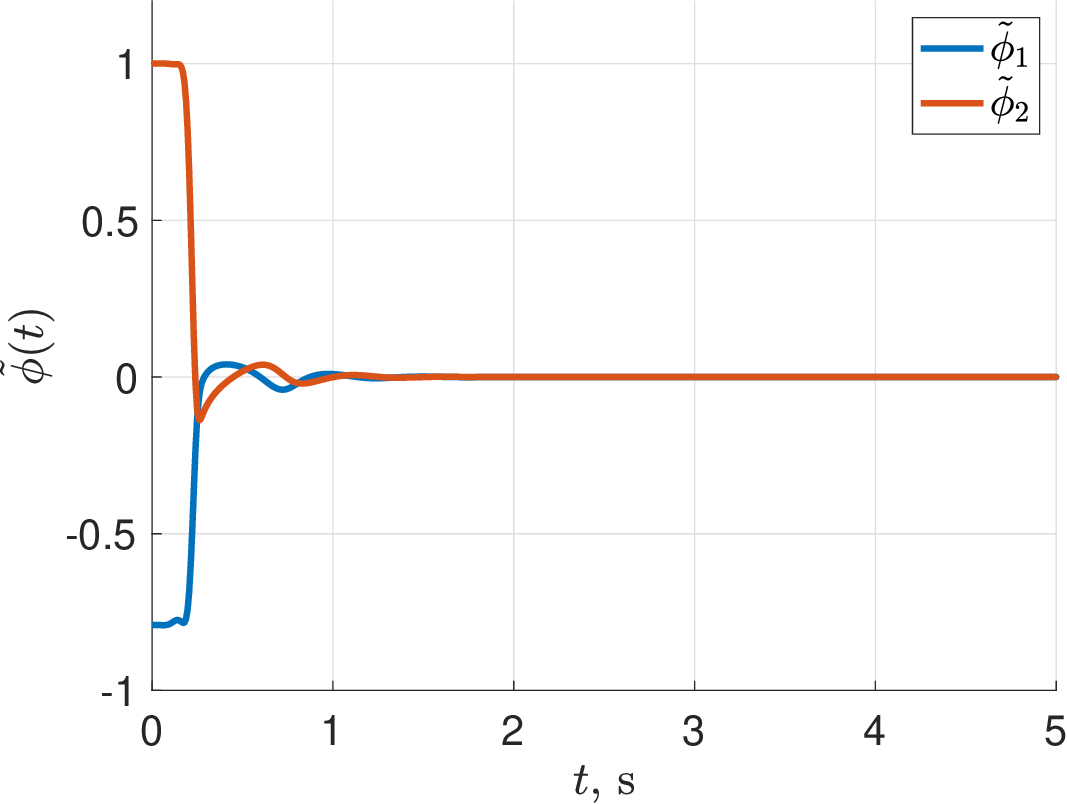}}
		\caption{\label{fig_2} Flux observation errors for the case of time-varying desired rotor speed}
\end{figure*}
%
%%%%%%%%%%%%%%%%
\subsection{Wound Rotor Induction Motor (WRIM)}
\lab{subsec32}
%%%%%%%%%%%%%%%%
%
In this section we design an ASLO for the flux of the WRIM. We make the important observation that in WRIM---in contrast to the more popular squirrel-cage IM---the rotor windings are connected to slip rings, hence the {\em rotor current is measurable}. Although WRIM are much less popular than squirrel-cage IM they are still used in high-power industrial systems like cranes and mine hoists, ball and grinding mills, conveyors and large pumps, hence the interest to design an ASLO for it.
%%%
\subsubsection{Mathematical model}
\lab{subsec321}
%%%%%%%%%%%%%%%5
%
The $\alpha\beta$ model describing the behavior of the IM is given by \cite[eqs. (10.1)-(10.4)]{ORTetal_bookel}
\begali{
\label{dotlam}
\dot  \phi & = -\begmat{ R_s {\bf I}_2 &  {\bf 0}_2 \\ {\bf 0}_2 & R_r {\bf I}_2 }i + \begmat{{\bf I}_2   \\ {\bf 0}_2 }u\\
\nonumber
\dot \theta &= \omega \nonumber \\
\nonumber
\dot \omega &=  - {R_m \over J} \omega +{1 \over J} i_s^\top{\mathcal J}  \phi_s - {1 \over J}\tau_L
}
where  $ \phi(t) =[  \phi_s^\top(t) \;  \phi_r^\top(t)]^\top\in\rea^4$ is the flux vector,  $i(t)= [i_s^\top(t) \; i_r^\top(t) ]^\top$ is the current vector, $u(t) \in \rea^2$ is the stator voltage,  $\theta(t) \in [0,2\pi)$ is the rotor angle and $ \omega(t) \in \rea$ is the rotor angular velocity. The constants $R_r$ and $R_s$ are the rotor and stator resistances, respectively, $J$ is the rotor moment of inertia, $R_m$ the friction coefficient, $\tau_L$ is the load torque, and  $\tau(t) \in \rea$ is the electrical torque, where we defined the rotation matrix
$$
{\mathcal J}:= \left[\begin{array}{cc} 0 & -1 \\ 1 & 0 \end{array} \right].
$$

Under a linear magnetic assumption, fluxes and currents are related through the $4 \times 4$ inductance matrix of the windings as \cite[eq. (10.5)]{ORTetal_bookel}
\begin{equation}
\label{lambda}
 \phi=\left [\begin{array}{cc} L_s {\bf I}_2  & L_{sr} {\bf e}^{{\mathcal J} \theta}\\
L_{sr} {\bf e}^{-{\mathcal J} \theta} & L_r {\bf I}_2 
 \end{array} \right] i,
\end{equation}
where $L_r$, $L_s$ are self-inductances, $L_{sr}$ is the mutual inductance, with the rotation matrix
$$
{\bf e}^{{\mathcal J}\theta}= \left[\begin{array}{cc} \cos(\theta) &  -\sin(\theta) \\ \sin(\theta) & \cos(\theta)\end{array}\right].
$$

Similarly to the case of the PMSM the task is to design an ASLO for the flux assuming measurable only the current and the voltage vectors. Hence, we define the output signals $y(t) \in \rea^4$ as\footnote{We recall the important observation made above that we assume the rotor currents are {\em measurable}---as it is the case in WRIM.}
\begequ
\lab{ywrim}
y=\begmat{y_s \\ y_r}:=\begmat{i_s \\ i_r}.
\endequ
To simplify the notation we introduce the following convention: we use the subindex $\kappa$ to denote either a stator or a rotor signal. That is, we denote $(\cdot)_\kappa,\;\kappa\in\{s,r\}$.
%   
%%%%%%%%%%%%%%%%%
\subsubsection{Main Result}
\lab{subsec322}
%%%%%%%%%%%%%%
%
\begpro 
\lab{pro3}
Consider the dynamics of the IM system given in \eqref{dotlam}  and the filter \eqref{fil}. Define the following {\em measurable} signals \vspace{-0.2cm}
\begalis{
z_{s_i}&:=- R_s i_{s_i}+ u_i, \;z_{r_i}:=- R_r i_{r_i},\;i=1,2,
}
and introduce the dynamic extension  \vspace{-0.2cm}
\begalis{
w_{\kappa_4}& ={1 \over \lambda L_s}\calf[z_{\kappa_1}] -{p \over \lambda}\calf[y_{\kappa_1}]  \\
w_{\kappa_5}&= {1 \over  \lambda L_s}\calf[z_{\kappa_2}] - {p \over \lambda}\calf[y_{\kappa_2}]\\
w_{\kappa_6}&=  {1 \over \lambda }\calf[z_{\kappa_1}\calf[y_{\kappa_1}]]+{1 \over \lambda }\calf[z_{\kappa_2}\calf[y_{\kappa_2}]]+ {1 \over \lambda L_s}\calf[z_{\kappa_1}{1 \over \lambda}\calf[z_{\kappa_1}]] \\
&\;\;+{1 \over \lambda L_s}\calf[z_{\kappa_2}{1 \over \lambda}\calf[z_{\kappa_2}]] -{p \over \lambda}\calf[z_{\kappa_3}] \\
w_{\kappa_7}&=\calf[w_{\kappa_4}]\\
w_{\kappa_8}&=\calf[w_{\kappa_5}]\\
w_{\kappa_9}&=\calf[w_{\kappa_6}]+{1 \over \lambda}\calf[z_{\kappa_1}\calf[w_{\kappa_4}]] +{1 \over \lambda}\calf[z_{\kappa_2}\calf[w_{\kappa_5}]],
}

Define the signals 
\begalis{
\Delta_\kappa:=w_{\kappa_4}w_{\kappa_8}-w_{\kappa_5}w_{\kappa_7},
}
which we assume are different to zero for all $t \geq 0$.
 
\noindent {\bf P1} The ASLO of $\phi_\kappa$ is given by
\begequ
\lab{www}
\phi_\kappa(t)={1 \over \Delta_\kappa(t)} \begmat{w_{\kappa_6}(t) w_{\kappa_8}(t)-w_{\kappa_5}(t)w_{\kappa_9}(t)\\ w_{\kappa_4}(t) w_{\kappa_9}(t)-w_{\kappa_6}(t)w_{\kappa_7}(t) }.
\endequ
\noindent {\bf P2} Define the A-ASLO 
\begin{align}
	\dot {\hat \phi}_{\kappa_1} &=z_{\kappa_1}-\gamma_{\kappa}\hat \phi_{\kappa_1}+ \gamma_\kappa {1 \over \Delta_{\kappa}}(w_{\kappa_6} w_{\kappa_8}-w_{\kappa_5}w_{\kappa_9}) \nonumber\\
	\dot  {\hat \phi}_{\kappa_2} &=z_{\kappa_2}-\gamma_\kappa \hat \phi_{\kappa_2}+ \gamma_\kappa {1 \over \Delta_\kappa}(w_{\kappa_4} w_{\kappa_9}-w_{\kappa_6}w_{\kappa_7}),
	\lab{asyobsphir}
\end{align}
where $\gamma_\kappa >0$ is a tuning gain. Then, the observation errors $\tilde \phi_\kappa:=\hat \phi_\kappa-\phi_\kappa$ satisfy
$\dot {\tilde \phi}_\kappa = -\gamma_\kappa {\tilde \phi}_\kappa$.
\endpro

%%%%%%%%%%%%%%%%
\subsection{Discussion}
\lab{subsec323}
%%%%%%%%%%%%%%%%
%
We make the following observations regarding the application of ASLO to electrical machines.

\noindent {\bf O5} A model for the generalized electric motor (given in Euler-Lagrange (EL) form) may be found in  \cite[eqs. (9.6)-(9.8)]{ORTetal_bookel}, see also \cite{LIUetal} where many practical examples are discussed in detail. The electrical part may be expressed in terms of the flux $\phi(t) \in \rea^{n_e}$ and the current $y(t) \in \rea^{n_e}$, with $n_e \in \intnum_+$ the number of windings (in stator and rotor), as follows: 
\begalis{
\dot \phi&=-Ry+Gu\\
\phi&=L(\theta)y+\mu(\theta),
}
with $\theta(t) \in [0,2\pi)$ the rotor angle and $G\in \rea^{n_e \times m}$ a constant input matrix. The first equation is, clearly, Faraday's law.  From the derivations given in the proofs of Propositions \ref{pro2} and \ref{pro3}, we see that there are two key steps for the design of an ASLO:
\begin{itemize}
\item [{\bf S1}] The establishment of {\em linear algebraic relations} of the form \eqref{algsys}.
\item [{\bf S2}] The requirement that the matrices $\cala(t)$ and $\calb(t)$ should be {\em measurable}.
\end{itemize}
Towards this end, we exploited the structure of the mappings $L(\theta)$ and $\mu(\theta)$ to, in both examples, eliminate from the algebraic relations the dependence on $\theta$. Moreover, in the IM example we needed to restrict ourselves to the case of wounded rotor IM because---to fulfill the requirement of step {\bf S2}---it was necessary to assume measurement of the rotor current, which is unavailable in squirrel cage IM.  

\noindent {\bf O6}  There are many practical motor examples that fulfill the requirements above---{\em e.g.}, the PM stepper motor \cite[Example 9.9]{ORTetal_bookel}.

\noindent{\bf O7}  Although it is clear from \eqref{flucur} that, with the knowledge of the flux $\phi$, the rotor angle in the PMSM can be computed  via
$$
\theta={1 \over n_p}\arctan\Big\{ {Ly_2 - \phi_2 \over Ly_1 - \phi_1}\Big\},
$$
it is interesting to note that it is possible to design an ASLO {\em directly for $\theta$}. The key steps are reported in \cite{PYRetal} where an adaptive observer based on the Dynamic Regressor Equation and Mixing parameter estimation procedure \cite{ARAetal_tac17} is proposed. The key step is the derivation of the algebraic relation  \cite[eq. (28)]{PYRetal}, which is {\em linear} in the vector $\col(\cos(n_p \theta),\sin(n_p \theta))$ (denoted $\chi$ in  \cite{PYRetal}). Although this relation is scalar, it is possible to generate another scalar relation via application of the SL---as it is done here. However, in  \cite{PYRetal} the interesting variation of deriving the second relation using a {\em different filter} in the SL is explored. 

\noindent {\bf O8} One particularly important motor, which has a very wide range of applications, is the {\em interior} PMSM \cite{NAMbook}. For this motor, due to the presence of higher-order harmonics, we have that the inductance matrix $L(\theta)$ takes the complicated form:
$$
L(\theta)={L_d + L_q \over 2} {\bf I}_2+  {L_d - L_q \over 2}\left[\begin{array}{cc} \cos(2\theta) &  \sin(2\theta) \\ \sin(2\theta) & -\cos(2\theta)\end{array}\right],
$$  
where $L_d,L_q$ are the $dq$ stator inductances---hence the construction indicated in observation {\bf O5} does not seem feasible. In spite of that, it is possible to design as ASLO for the {\em active flux} which is defined as
$$
{\bf x}:=\phi-L_q y.
$$
The basic step is carried out in \cite[eq. (21)]{CHOetal} where the authors present a (perturbed) algebraic relation, which is linear in the vector ${\bf x}$. Similarly to \cite{PYRetal}, this relation is {\em scalar}, but as explained above, it is possible to generate another scalar relation via application of the SL or using a {\em different filter} in the SL. The {\em perturbation} is the sum of two terms, denoted $\epsilon(t)$ and $d(t)$ in \cite{CHOetal}, the first is the exponentially convergent contribution of the filters' initial conditions, while $d(t)$ is proven to be small due to practical considerations. Henceforth, both terms can be neglected. One interesting feature of the developments of \cite{CHOetal} is that a {\em new version} of the SL---necessary to carry out the analysis---is established in \cite[Lemma 2]{CHOetal}. In \cite{CHOetal} an ASLO {\em is not derived}, instead a gradient descent estimation algorithm is implemented to asymptotically reconstruct ${\bf x}$. It is worth mentioning that the stability analysis carried out in \cite{CHOetal} imposes a hard constraint on the estimation gain---styimieing the achievement of high performance. This obstacle has been removed in the recent paper \cite{YIetal}, where a new analysis technique based on the creation of a ``virtual" invariant manifold was utilized. It is interesting to note that it also possible to design an ASLO for the scheme reported in  \cite{YIetal}. 
%
%%%%%%%%%%%%%%%%
\section{Mechanical Systems Examples}
\lab{sec4}
%%%%%%%%%%%%%%%%
In this section we apply the ASLO procedure for mechanical systems. We consider the practical scenario where the {\em positions are measurable} and we want to find an ASLO for the coordinate {\em velocities}. First, we provide the solution for the two well-known examples of the Robotic Leg \cite{BULLEWbook} and the Ball-and-Beam (B\&B) \cite{HAUSASKOK} systems. Then, we give a general characterization of a class of mechanical systems amenable for the design of an ASLO. For the two examples we also give an A-ASLO. \vspace{-0.4cm}
\subsection{Robotic Leg system}
\label{sec41}
We consider in this section the Robotic Leg system, which was studied in \cite{BULLEWbook}. It has three DoF $q_1(t)\in \rea$ and $(q_2(t),q_3(t))\in (0,2\pi] \times (0,2\pi]$, and two control forces $u(t) \in \rea^2$.  The inertia and input matrices are of the form 
$$
M(q_1)=\begmat{ m_1 & 0 & 0\\ 0 & m_1 q_1^2 & 0 \\ 0 & 0 & m_2}, \quad G=\begmat{1 & 0 \\ 0 & 1\\ 0 & -1 },
$$
with positive numbers $m_1$ and $m_2$. Besides, there is no potential energy term nor friction matrix. Consequently, the EL equations of motion take the form
\begin{align}
\nonumber
m_1 \ddot q_1 =& m_1 q_1 \dot q_2^2 +u_1,\;m_1 q_1^2 \ddot q_2 =-2 m_1 q_1 \dot q_1 \dot q_2 + u_2, \\
m_2 \ddot q_3 =&-u_2.
\label{rlsys}
\end{align}

As indicated above, it is assumed that the positions $q$ are available for measurement and our objective is to design an ASLO for the corresponding velocities $\dot q$.
\begpro 
\lab{pro4}
Consider the dynamics of the Robotic Leg system given in \eqref{rlsys}. Define the signals\vspace{-0.1cm}
\begali{
\lab{zrl}
z_1&:=e^{2 \ln(q_1)},\;z_2:={u_2 \over m_1 q_1^2}
}
and introduce the {\em dynamic extension}
\begali{
w_1& =p\calf[q_2]+{1 \over \lambda}\calf\Big[z_1 z_2  \calf\Big[{1 \over z_1}\Big]\Big],\;w_2={1 \over m_2 \lambda}\calf[u_2],\\
w_3&=p\calf[q_3],\;w_4={1 \over \lambda}\calf\Big[q_1\Big({w_1 \over z_1 \calf[z_1]}\Big)^2+{u_1 \over m_1}\Big],
\lab{dynrl}
}
with $\calf(p)$ the LTI filter defined in \eqref{fil}.

\noindent {\bf P1} The ASLO for $\dot q$ is given by 
%%%%%%%
\begin{equation*}
{\small
\hspace{-.5cm} \dot q(t)=\col\Big(p\calf[q_1]+ w_4(t),\;  {w_1(t) \over   z_1(t) \calf[z_1(t)]},\;w_2(t)+w_3(t)\Big).
}
\end{equation*}

\noindent {\bf P2} Denote the estimate of $\dot q$ as $\hat \omega$ and define an A-ASLO for $\dot q$ as
\begin{align*}
\hspace{-.6cm} \dot {\hat \omega}= &\col\Big( q_1 \dot q_2^2 +{1 \over m_1} u_1, \; -{2 \over q_1}\dot q_1 \dot q_2 + {1 \over m_1 q_1^2} u_2,\; -{1 \over m_2} u_2\Big) \\
& -\gamma \hat \omega+ \gamma \dot q,
\end{align*}
where  $\dot q$ in the right hand side of the equations above is replaced by its ASLO expressions given in {\bf P1} above and $\gamma>0$ is a tuning gain. Then, the speed observation error $\tilde \omega:=\hat \omega- \dot q$ satisfies $\dot {\tilde \omega} = -\gamma {\tilde \omega}.$
\endpro
\subsection{The ball and beam system}
\lab{subsec41}
In this section we study the B\&B system. It has two DoF,  $q_1(t) \in (-{\ell \over 2},{\ell \over 2})$ and $q_2(t) \in (0,2\pi]$, with $\ell \in \rea_+$ the length of the beam---see \cite{HAUSASKOK} for further details. The inertia matrix is $M(q_1)=\diag\{  1, \ell^2 +q_1^2 \}$. The potential energy is $V(q)=gq_1\sin(q_2)$ and the input matrix is $G=\col(0, 1)$, furthermore we assume there is no friction. The EL equations of motion take the form\vspace{-0.2cm}
\begali{
\ddot q_1&=q_1\dot q_2^2-g \sin(q_2), \nonumber \\
( \ell^2+q^2_1)\ddot q_2&=-2q_1\dot q_1 \dot q_2 -g q_1\cos(q_2)+u
\lab{bbsys}
}
with $u(t) \in \rea$ the control torque.
%%%%%%%%%%%
\begpro 
\lab{pro5}
Consider the dynamics of the B\&B system given in \eqref{bbsys}. Define the following {\em measurable} signals
\begalis{
z_1 &:= {1 \over  \ell^2+q^2_1}\Big[-g q_1\cos(q_2)+u\Big],\;z_2:=e^{\ln(\ell^2+q^2_1)}
}
and introduce the {\em dynamic extension}
 \begalis{
w_1 &:= \calf[z_2^{-1}],\;w_2= \calf[w_1z_1z_2]
} 
\noindent {\bf P1} The ASLO of $\dot q$ is given by
\small{
$$
\dot q(t)=\col\Big( p\calf[q_{1}]+{1 \over \lambda}\calf[ q_{1}\dot q_2^2-g \sin(q_2)],{1 \over w_1 z_2}\Big({1 \over \lambda}w_2+p\calf[q_2]\Big)\Big).
$$
}
%%%%%%
{\bf P2} Denote the estimate of $\dot q$ as $\hat \omega$ and define an A-ASLO for $\dot q$ as
\begin{align*}
\hspace{-1.3cm} \dot {\hat \omega}=& \col\Big(q_1\dot q_2^2-g \sin(q_2),\; \frac{-2q_1\dot q_1 \dot q_2 -g q_1\cos(q_2)+u}{\ell^2+q^2_1} \Big) \\
& -\gamma \hat \omega+ \gamma \dot q,
\end{align*}
where  $\dot q$ in the right hand side of the equations above is replaced by its ASLO expressions given in {\bf P1} above and $\gamma>0$ is a tuning gain. Then, the speed observation error $\tilde \omega:=\hat \omega- \dot q$ satisfies $\dot {\tilde \omega} = -\gamma {\tilde \omega}.$ \vspace{-0.2cm}
\endpro	
\subsection{Identification of a class of mechanical systems amenable for ASLO design}
\label{subsec43}
%%%%%%%%%%%%%%%
We consider general  mechanical systems whose dynamics is described by the well-known EL equations of motion \cite{ORTetal_bookel}
\begin{equation}
\label{EL}
M(q)\ddot q +C(q,\dot q)\dot q + R \dot q  + \nabla V(q)=G(q) u,
\end{equation}
where $q(t) \in \rea^n$ are the configuration variables, $u(t) \in \rea^{m}$ are the control signals,  $M: \rea^n \rightarrow \rea^{n \times n}$, is the positive definite
generalized inertia matrix, $C(q, \dot q) \dot q $ are the Coriolis and centrifugal forces, with $C:  \rea^n \times \rea^n \rightarrow \rea^n \times \rea^n $ defined as 
\begin{equation}
\label{Cor}
C(q, \dot q)= \left[ \nabla_q\left(M(q) \dot q\right) -\frac{1}{2} [\nabla_q \left(M(q) \dot q\right)]^\top \right] \dot q, 
\end{equation}
$V : \rea^n \rightarrow \rea$ is the systems potential energy,  $G: \rea^n \to \rea^{n \times m}$ is the full rank input matrix and $R\in \rea^{n \times n}$, which satisfies $R=R^\top \geq 0$, is the friction matrix.
  
To simplify the definition of the class,  we partition the generalized coordinates as $q=\col (q_1, q_2) $, with $q_1 \in \rea^s$ and $q_2 \in \rea^{m}$ and $s:=n-m$. Similarly, the inertia matrix and input matrix are partitioned as
$$
 M(q)=\left[ \begin{array}{cc} m_1(q)  & m^{\top}_{2}(q) \\ m_{2}(q)  & m_3 (q)   \end{array} \right], \quad  G(q)= \left[ \begin{array}{c} g_{1}(q)   \\ g_2 (q)   \end{array} \right],
$$
where $m_1 : \rea^{n} \rightarrow \rea^{s \times s}$, $m_2 : \rea^n \rightarrow \rea^{m \times s}$,  $m_3 : \rea^n \rightarrow \rea^{m\times m}$, $g_1: \rea^n \rightarrow \rea^{s \times n}$ and $g_2: \rea^n \rightarrow \rea^{m \times n}$. In addition, we assume that $R=\diag \{R_1, R_2\}$ with $R_1\in \rea^{s \times s}$ and $R_2 \in \rea^{m \times m}$ {\em diagonal}, and $u=\col(u_1, u_2)$, with $u_1(t) \in \rea^{s}$ and $u_2(t) \in \rea^{m}$

Using this notation the class is identified imposing two assumptions. The first one, which is quite restrictive, is the following one.
%%%%%%%%%%%%%
\begin{assumption} 
\label{ass2}
The inertia matrix  is {\it diagonal} and depends only on the variables $q_1$. Moreover, the sub-block matrix $m_1$  is {\it constant}. Consequently, $M(q)$ becomes:
$$
M(q_1)= \left[ \begin{array}{cc} m_1  & {\bf 0}_{s \times m}  \\ {\bf 0}_{m \times s}  & m_3 (q_1)   \end{array} \right].
$$
\end{assumption}

We have decided to make the diagonal components of the inertia matrix {\em constant} for the upper part ones, and for those in the lower part, they are {\em dependent on $q_1$ only}. This choice was made to capture the two examples studied before. As will be shown below it would have been possible to swap the roles and make the upper terms dependent on $q_2$ and the lower ones constant. In the first case, we will design our ASLO for $\dot q_2$ first---as in the previous two examples. On the other hand, with the second choice it will be an ASLO for $\dot q_1$ the one designed first.

Under Assumption \ref{ass2}, we have that 
$$
\nabla_q\left(M(q) \dot q\right)= \left[\begin{array}{cc}{\bf 0}_{m \times m}  & {\bf 0}_{m \times s} \\ \nabla_{q_1}\left[m_3(q_1) \dot q_2\right] & {\bf 0}_{s \times s} \end{array} \right],
$$
and the Coriolis forces \eqref{Cor} take the form
\begin{align}
\label{Cqs}
C(q, \dot q) \dot q= \left[\begin{array}{c}-\frac{1}{2} \nabla_{q_1}\left[m_3(q_1) \dot q_2\right]^\top \dot q_2 \\  \nabla_{q_1}\left[m_3(q_1) \dot q_2\right] \dot q_1 \end{array} \right].
\end{align}
Invoking Assumption \ref{ass2} and \eqref{Cqs} the EL model \eqref{EL} is
\begin{align}
\label{EL2}
m_1 \ddot q_1-\frac{1}{2}\nabla_{q_1}  \left[m_3(q_1) \dot q_2\right]^\top \dot q_2 +R_1 \dot q_1+ \nabla_{q_1} V(q)=& g_1(q) u \nonumber \\
m_3(q_1) \ddot q_2 + \nabla_{q_1}\left[m_3(q_1) \dot q_2\right] \dot q_1 +R_2 \dot q_2 +\nabla_{q_2} V(q)= & g_2(q)u.
\end{align}

Exploiting the fact that $m_3$ is a diagonal matrix, and after some straightforward calculations, the $j$-th component of the second term in the right hand side of the second equation of \eqref{EL2}  can be written as:
$$
 \bfe^\top_j\nabla_{q_1}\left[m_3(q_1) \dot q_2\right] \dot q_1=\left[ \sum_{i=1}^{s} \Big( \nabla_{q_{1_i}} m_{3_j}(q_1)  \Big)\dot q_{1_i} \right] \dot q_{2_j},j\in\bar{m}.  
$$
Since $R_2$ is diagonal, the $j$-term of the second equation in \eqref{EL2} takes the form
$$
\ddot q_{2_j}+{1 \over m_{3_j}(q_1)}\left[ \sum_{i=1}^{s} \Big( \nabla_{q_{1_i}} m_{3_j}(q_1)  \Big)\dot q_{1_i} \right] \dot q_{2_j} +\bar r_{2_j} (q_1) \dot q_{2_j} =: z_{0_j}(q),
$$
where we defined
\begalis{
\bar r_{2_j}(q_1) &:={1 \over m_{3_j}(q_1)}  R_{2_j}\\
z_{0_j}(q) &:={1 \over m_{3_j}(q_1)} \bfe^\top_j\Big[ g_2(q)u - \nabla_{q_2} V(q)\Big].
}

To define the class of systems addressed in this note, the following final assumption is required.

\begin{assumption}
\label{ass3}
There exist functions $z_j : \rea^s \rightarrow \rea$, such that
\begin{equation}
\dot z_j(q_1) =\frac{1}{m_{3_j}(q_1)}\left[  \sum_{i=1}^{s}  \Big( \nabla_{q_{1_i}} m_{3_j}(q_1)  \Big)\dot q_{1_i} \right]+ \bar r_{2_j}(q_1),\;j \in \bar{m}.
\end{equation}
Since $z_j(q_1)$ depends solely on $q_1$, it is {\em measurable}. 
\end{assumption} 

Finally, Assumption  \ref{ass3} allows us to express the $m$ terms of the vector $\ddot q_2$ in the following form:
\begequ
\lab{ddotq2} 
\ddot q_{2_j}+\dot z_j(q_1) \dot q_{2_j} = z_{0_j}(q_1), \quad j \in \bar m. 
\endequ
Notice that in both examples above Assumptions \ref{ass2} and \ref{ass3} are satisfied---with $\bar r_{2_j}(q_1) \dot q_{2_j}=0$. The latter can be verified comparing \eqref{ddotq2} with \eqref{for1} and \eqref{for2}.\\

Once the acceleration vector $\ddot q_2$ is expressed in the form \eqref{ddotq2}  the design of the ASLO for $\dot q$ follows the construction done in the proofs of Proposition \ref{pro4} and \ref{pro5}, and is omitted for brevity. 
%
%%%%%%%%%%%%%%%
\section{Conclusions and Future Work}
\lab{sec5}
%%%%%%%%%%%%%%
%
We presented in this paper a radically new state observer design where the objective is to obtain an {\em algebraic} relation between the observed state and nonlinear combinations of the systems inputs and outputs and their filtered versions. The main requirement is that this algebraic relation should hold true for all $t \geq 0$---which should be contrasted with standard asymptotic (or fixed/finite time) observers. The methodology is applicable to nonlinear systems where the {\em derivative} of some of the states is available for measurement, as it is often the case in many physical systems. The key technical tool is the SL that ``extracts" (via filtering) from the product of two signals, the derivative of one of them. The applicability of the method is illustrated with the physical examples of motors and mechanical systems, yielding radically new, high performance, observer designs, even in the case of LTI systems.  

The introduction of this new ASLO opens up a wide research avenue for the state observation field with many interesting questions to be explored, among them:  

\bul Assessment of the {\em robustness} properties of ASLO---comparing with classical asymptotically convergent versions. 

\bul It is argued in the paper that, in contrast with standard observer designs, ASLO does not require an explicit assumption on {\em observability}. However, some operations on the construction of ASLO---for instance the division by the signal $\Delta$ in the motor examples---may create numerical problems due to the appearance of singularities. A better understanding of this phenomenon is in dire need.    

\bul Exploration of its applicabilty to {\em other physical systems}. It should be underscored that we have recently succesfully applied it to {\em active magnetic bearing} systems \cite{MASbook}---a result that will be reported shortly.

\bul Evaluate the effect in the observers performance of choosing different filter time-constants, {\i.e.},  {\em different} $\lambda$.

\bul Replace the SL by the use of {\em different filters} to generate the required algebraic equations---as done in \cite{PYRetal} in the context of adaptive observer design.  

\bul In the spirit of the robustness to constant output disturbances indicated in Observation {\bf O5} in Subsection\ref{subsec25}, use that fact that the SL is valid for any LTI filter to investigate the possibilty of rejecting disturbances with known {\em internal model} \cite{ISIbook}: for instance $(p^2 + \omega^2)[\sin(\omega t)]=0$.

\bul Study the possibility of replacing the classical SL with the {\em nonlinear version} reported in \cite{KRSKOK}.
%
%%%%%%%%%%%%

%%%%%%%%%%%%%%%%%%
\appendix

%

%%%%%%%%%%%%%%
\section{Swapping Lemma \cite[Lemma 3.6.5]{SASBODbook}}
\lab{appa}
%%%%%%%%%%%%%%%
%
Consider the stable LTI filter \eqref{fil}. Given the signals $w:\rea_+ \to \rea^{n \times m}$ and $v:\rea_+ \to \rea^m$, with $v(t)$ differentiable, the following identity holds true\footnote{For ease of presentation we have selected the basic filter \eqref{fil}. However, as shown in \cite[Lemma 3.6.5]{SASBODbook} the lemma holds true of an arbitrary LTI system. See also the footnote in the observation {\bf O3} of Subsection \ref{subsec23}.} 
$$
\calf[w v]= \calf[w] v-  \calf\Big[ {1 \over \lambda}\calf [w]\dot v\Big].
$$
\vspace{-0.8cm}
%
%
%%%%%%%%%%%%%%%%%%%%
\section{Proof of Proposition \ref{pro2}}
\lab{appb}
%%%%%%%%%%%%%%%%%%%
% 
The main objective of the proof is to establish a linear algebraic relation of the form
\begequ
\lab{algsys}
\cala(t) \phi(t) = \calb(t),
\endequ
where $\cala(t) \in \rea^{2 \times 2}$ and $\calb(t) \in \rea^2$ are measurable matrices, with $\cala(t)$ invertible, this will be possible via succesive applications of the well-known SL given in \cite[Lemma 3.6.5]{SASBODbook}, and recalled in Appendix \ref{appa} for ease of reference.

First, notice that using the well-known fact that 
$$
\cos^2(n_p x_3)+\sin^2(n_p x_3)=1,
$$
we get from \eqref{flucur} an alternative algebraic equation, which is the main step to design an {\it algebraic} observer
$$
|Ly- \phi|^2=\lambda_m^2.
$$
Expanding it, we get 
$$
L^2 y_1^2 -2 L y_1 \phi_1 + \phi_1^2+L^2 y_2^2 -2L y_2 \phi_2 + \phi_2^2=\lambda_m^2,
$$
which can be written as
\begin{equation}
	y_1  \phi_1 + y_2  \phi_2 =\frac{1}{2L}|\phi|^2 +z_3
	\label{y12}
\end{equation}
with the measurable signal $z_3$ defined in \eqref{z}. Applying the SL to each of the terms of \eqref{y12} we obtain
\begalis{
	\calf[y_1  \phi_1] &=\calf[ \phi_1 y_1]= \phi_1\calf[y_1]-{1 \over \lambda}\calf[z_1\calf[y_1]]\\
	\calf[y_2  \phi_2] &=\calf[ \phi_2 y_2]= \phi_2\calf[y_2]-{1 \over \lambda}\calf[z_2\calf[y_2]]\\
		\calf[1 \times |\phi|^2]&=\calf[1 \times  \phi^2_1] + \calf[1 \times \phi^2_2]
}
\begalis{
		&= \phi^2_1\calf[1]-{1 \over \lambda}\calf[(2 \phi_1 \dot  \phi_1)\calf[1]] + \phi^2_2\calf[1]-{1 \over \lambda}\calf[(2 \phi_2 \dot  \phi_2)\calf[1]] \\
		&=  |\phi|^2 - {2 \over \lambda}\Big(\calf[ \phi_1z_1]  + \calf[ \phi_2 z_2]\Big) \\
		& =  |\phi|^2 -{2 \over \lambda}\Big( \phi_1\calf[z_1]-{1 \over \lambda}\calf[z_1\calf[z_1]] + \phi_2\calf[z_2]-{1 \over \lambda}\calf[z_2\calf[z_2]]\Big).
		}
Thus after applying filter $\calf$ to \eqref{y12}, and using the equations above, we obtain
\begali{
\nonumber
		&  \phi_1\calf[y_1]-{1 \over \lambda}\calf[z_1\calf[y_1]]  + \phi_2\calf[y_2]-{1 \over \lambda}\calf[z_2\calf[y_2]]\\
\lab{fff}
		&=\frac{1}{2L} \Big[ |\phi|^2 -{2 \over \lambda}\Big( \phi_1\calf[z_1]-{1 \over \lambda}\calf[z_1\calf[z_1]]+ \phi_2\calf[z_2] \nonumber \\
		&\;\;-{1 \over \lambda}\calf[z_2\calf[z_2]]\Big) \Big]+\calf[z_3].
}

To simplify the notation let us define variables:
\begalis{
	w_1 &:= \calf[y_1]+{1 \over \lambda L}\calf[z_1] \\
	w_2 &:= \calf[y_2]+{1 \over \lambda L}\calf[z_2] \\ 
		w_3 &:=  {1 \over \lambda }\Big[\calf[z_1\calf[y_1]]+\calf[z_2\calf[y_2]] \\
		&\;\;\;+{1 \over \lambda L}\Big(\calf[z_1\calf[z_1]]+\calf[z_2\calf[z_2]]\Big)\Big] +\calf[z_3].
}
Notice that these three signals, $w_1, w_2$ and $w_3$, are {\em measurable}.

Using the definitions above in \eqref{fff} we obtain
\begin{equation}
	\label{y3}
 \phi_1 w_1+ \phi_2 w_2= \frac{1}{2L}  |\phi|^2+w_3.
\end{equation}
Now, we can write the system of two equations \eqref{y12} and \eqref{y3} as \vspace{-0.5cm}
\begali{
\nonumber
	y_1  \phi_1+y_2  \phi_2 = \frac{1}{2L} |\phi|^2  + z_3\\
\lab{algeqspmsm}
	w_1  \phi_1 + w_2 \phi_2 = \frac{1}{2L}  |\phi|^2 +w_3 
}
After subtraction we obtain: 
\begequ
\lab{w6}
  w_4  \phi_1 +w_5  \phi_2 = w_6 
\endequ
where we defined 
$$
w_4:= w_1 - y_1,\;w_5:= w_2-y_2,\;w_6= w_3-z_3
$$
which are, clearly, {\em measurable}.

Now we can apply filter $\calf$ to the equation above:
$$
\calf[w_6]=\calf[ \phi_1w_4]+\calf[ \phi_2w_5]
$$
Applying the SL we can rewrite this equation as
$$
\calf[w_6]= \phi_1\calf[w_4]-{1 \over \lambda}\calf[z_1\calf[w_4]]+ \phi_2\calf[w_5]-{1 \over \lambda}\calf[z_2\calf[w_5]],
$$ 
which can be compactly rewritten as
\begin{equation}
\lab{w9}
  w_7\phi_1 +  w_8\phi_2 =w_9
\end{equation}
where we defined
\begalis{
w_7 &:=\calf[w_4],\;w_8 :=\calf[w_5] \\
w_9&:=\calf[w_6]+{1 \over \lambda}\calf[z_1\calf[w_4]] +{1 \over \lambda}\calf[z_2\calf[w_5]].
}
With the definitions
$$
\cala(t):= \begmat{w_4 & w_5 \\ w_7 & w_8},\;\calb(t) :=\begmat{ w_6 \\ w_9}.
$$
the equations \eqref{w6} and \eqref{w9} constitute the algebraic system \eqref{algsys} we are looking for. 

Under the {\em generic assumption} that the matrix determinant 
$$
\Delta(t)=w_4(t)w_8(t)-w_5(t)w_7(t) \neq 0,\;\forall t \geq 0,
$$
we can solve this system to finally get the desired ASLO:
\begin{equation}
	\label{y56}
\begmat{\phi_1\\ \phi_2}={1 \over \Delta} \begmat{w_6w_8-w_5w_9 \\ w_4w_9 -w_6w_7},
\end{equation}
proving claim {\bf C1}.

Finally, to prove the claim {\bf C2} notice that, according to the definition of $z$ in \eqref{z}, we have that $(z_1,z_2)=\dot \phi$. Replacing this in  \eqref{asyobsphi} yields \eqref{errequphi}.
\qed
%
%%%%%%%%%%%%%%%%%%%%
\section{Proof of Proposition \ref{pro3}}
\lab{appc}
%%%%%%%%%%%%%%%%%%%
% 
Similarly to the proof of Proposition \ref{pro2}, the main objective of the proof is to establish, in this case {\em two}, linear algebraic relations of the form
\begequ
\lab{algsys1}
\cala_\kappa(t) \phi_\kappa(t) = \calb_\kappa(t),\; \kappa \in \{s,r\},
\endequ
where $\cala_\kappa(t) \in \rea^{2 \times 2}$ and $\calb_\kappa(t) \in \rea^2$ are {\em measurable} matrices, with $\cala_\kappa(t)$ {\em invertible}. 

We note that \eqref{lambda} can be decomposed as follows
 \begalis{
 \phi_s-  L_s i_s & = L_{sr} {\bf e}^{{\mathcal J}\theta} i_r \\
  \phi_r-  L_r i_r  &= L_{sr} {\bf e}^{-{\mathcal J}\theta}i_s. 
 } 
Hence, we have that
 \begali{
 \nonumber
 |\phi_s-  L_s i_s|^2 & = L^2_{sr} |i_r|^2 \\
 \lab{phisqu0}
  |\phi_r-  L_r i_r|^2  &= L^2_{sr} |i_s|^2, 
 }
where we used the fact that
$$
\Big({\bf e}^{{\mathcal J}\theta}\Big)^\top {\bf e}^{{\mathcal J}\theta}={\bf I}_2.
$$
Expanding the left hand side of both equations in \eqref{phisqu0} and rearranging terms we get that
\begin{align}
\lab{phiwrims}
\phi^\top_\kappa y_\kappa & ={1 \over 2 L_{\kappa}} |\phi_\kappa|^2+z_{\kappa_3}, 
\end{align}
where we use the notation \eqref{ywrim} and defined the {\it measurable} signals 
\begalis{
z_{s_3}:= \frac{L_s}{2}( i^2_{s_1} +    i^2_{s_2})  -\frac{L^2_{sr}}{2L_s} (i^2_{r_1} + i^2_{r_2}) \\
z_{r_3}= \frac{L_r}{2}( i^2_{r_1} +    i^2_{r_2})  -\frac{L^2_{sr}}{2L_r} (i^2_{s_1} + i^2_{s_2}). 
}
Now, we apply the SL to each of the terms in \eqref{phiwrims}.
\begalis{
	\calf[y_{\kappa_1}  \phi_{\kappa_1}] &=\calf[ \phi_{\kappa_1} y_{\kappa_1}]= \phi_{\kappa_1}\calf[y_{\kappa_1}]-{1 \over \lambda}\calf[z_{\kappa_1}\calf[y_{\kappa_1}]]\\
	\calf[y_{\kappa_2}  \phi_{\kappa_2}] &=\calf[ \phi_{\kappa_2} y_{\kappa_2}]= \phi_2\calf[y_{\kappa_2}]-{1 \over \lambda}\calf[z_{\kappa_2}\calf[y_{\kappa_2}]]
}
\begalis{
		&\calf[1 \times (\phi^2_{\kappa_1} +  \phi^2_{\kappa_2})]=\calf[1 \times ( \phi^2_{\kappa_1})] + \calf[1 \times ( \phi^2_{\kappa_2})]  \\
		&= \phi^2_{\kappa_1}\calf[1]-{1 \over \lambda}\calf[(2 \phi_{\kappa_1} \dot  \phi_{\kappa_1})\calf[1]]+ \phi^2_{\kappa_2}\calf[1]-{1 \over \lambda}\calf[(2 \phi_{\kappa_2} \dot  \phi_{\kappa_2})\calf[1]] \\
		&= |\phi_{\kappa}|^2 - {2 \over \lambda}\Big(\calf[ \phi_{\kappa_1}z_{\kappa_1}]  + \calf[ \phi_{\kappa_2} z_{\kappa_2}]\Big)
}
\small{
$$
 = |\phi_{\kappa}|^2  -{2 \over \lambda}\Big( \phi_{\kappa_1}\calf[z_{\kappa_1}]-{1 \over \lambda}\calf[z_{\kappa_1}\calf[z_{\kappa_1}]]+ \phi_{\kappa_2}\calf[z_{\kappa_2}]-{1 \over \lambda}\calf[z_{\kappa_2}\calf[z_{\kappa_2}]]\Big).
$$
}
Thus after applying filter $\calf$ to \eqref{phiwrims}, and using the equations above, we obtain
\small{
\begali{
\nonumber
		&  \phi_{\kappa_1}\calf[y_{\kappa_1}]-{1 \over \lambda}\calf[z_{\kappa_1}\calf[y_{\kappa_1}]]  + \phi_{\kappa_2}\calf[y_{\kappa_2}]-{1 \over \lambda}\calf[z_{\kappa_2}\calf[y_{\kappa_2}]]\\
\nonumber
		&=\frac{1}{2L_s} |\phi_{\kappa}|^2 -{2 \over \lambda}\Big( \phi_{\kappa_1}\calf[z_{\kappa_1}]-{1 \over \lambda}\calf[z_{\kappa_1}\calf[z_{\kappa_1}]]+ \phi_{\kappa_2}\calf[z_{\kappa_2}]\Big)\\
		& -{2 \over \lambda}\Big( -{1 \over \lambda}\calf[z_{\kappa_2}\calf[z_{\kappa_2}]]\Big) \Big]+\calf[z_{\kappa_3}]\Big).
\lab{fffs}
}
}
To simplify the notation let us define variables:
\begalis{
	w_{\kappa_1} &:= \calf[y_{\kappa_1}]+{1 \over \lambda L_s}\calf[z_{\kappa_1}] \\
	w_{\kappa_2} &:= \calf[y_{\kappa_2}]+{1 \over \lambda L_s}\calf[z_{\kappa_2}] \\ 
		w_{\kappa_3} &:=  {1 \over \lambda }\Big[\calf[z_{\kappa_1}\calf[y_{\kappa_1}]]+\calf[z_{\kappa_2}\calf[y_{\kappa_2}]]+\\
		&+{1 \over \lambda L_s}\Big(\calf[z_{\kappa_1}\calf[z_{\kappa_1}]]+\calf[z_{\kappa_2}\calf[z_{\kappa_2}]]\Big)\Big] +\calf[z_{\kappa_3}].
}
Notice that these three signals, $w_{\kappa_1}, w_{\kappa_2}$ and $w_{\kappa_3}$, are {\em measurable}.

Using the definitions above in \eqref{fffs} we obtain
\begequ
\lab{ppp}
 \phi_{\kappa_1} w_{\kappa_1}+ \phi_{\kappa_2} w_{\kappa_2}= \frac{1}{2L_s} ( \phi_{\kappa_1}^2 + \phi_{\kappa_2}^2)+w_{\kappa_3}.
\endequ
Now, we can write the system of two equations \eqref{phiwrims} and \eqref{ppp} as
\begalis{
	y_{\kappa_1}  \phi_{\kappa_1}+y_{\kappa_2}  \phi_{\kappa_2} &= \frac{1}{2L_s} |\phi_{\kappa}|^2  + z_{\kappa_3}\\
	w_{\kappa_1}  \phi_{\kappa_1} + w_{\kappa_2} \phi_{\kappa_2} &= \frac{1}{2L_s}  |\phi_{\kappa}|^2 +w_{\kappa_3}.
}
After subtraction we obtain: 
\begequ
\lab{ws6}
  w_{\kappa_4}  \phi_{\kappa_1} +w_{\kappa_5}  \phi_{\kappa_2} = w_{\kappa_6} 
\endequ
where we defined 
$$
w_{\kappa_4}:= w_{\kappa_1} - y_{\kappa_1},\;w_{\kappa_5}:= w_{\kappa_2}-y_{\kappa_2},\;w_{\kappa_6}:= w_{\kappa_3}-z_{\kappa_3}
$$
which are, clearly, {\em measurable}.

Now we can apply filter $\calf$ to the equation above:
$$
\calf[w_{\kappa_6}]=\calf[ \phi_{\kappa_1}w_{\kappa_4}]+\calf[ \phi_{\kappa_2}w_{\kappa_5}]
$$
Applying the SL we can rewrite this equation as
$$
\calf[w_{\kappa_6}]= \phi_1\calf[w_{\kappa_4}]-{1 \over \lambda}\calf[z_{\kappa_1}\calf[w_{\kappa_4}]]+ \phi_{\kappa_2}\calf[w_{\kappa_5}]-{1 \over \lambda}\calf[z_{\kappa_2}\calf[w_{\kappa_5}]],
$$ 
which can be compactly rewritten as
\begin{equation}
\lab{ws9}
  w_{\kappa_7}\phi_1 +  w_{\kappa_8}\phi_2 =w_{\kappa_9}
\end{equation}
where we defined
\begalis{
w_{\kappa_7} &:=\calf[w_{\kappa_4}],\;w_{\kappa_8} :=\calf[w_{\kappa_5}] \\
w_{\kappa_9}&:=\calf[w_{\kappa_6}]+{1 \over \lambda}\calf[z_{\kappa_1}\calf[w_{\kappa_4}]] +{1 \over \lambda}\calf[z_{\kappa_2}\calf[w_{\kappa_5}]].
}
With the definitions
$$
\cala_\kappa(t):= \begmat{w_{\kappa_4} & w_{\kappa_5} \\ w_{\kappa_7} & w_{\kappa_8}},\;\calb_\kappa(t) :=\begmat{ w_{\kappa_6} \\ w_{\kappa_9}}.
$$
the equations \eqref{ws6} and \eqref{ws9} constitute the algebraic system of \eqref{algsys1} we are looking for. 

Under the {\em generic assumption} that the matrix determinant satisfies
$$
\Delta_\kappa(t)=w_{\kappa_4}(t)w_{\kappa_8}(t)-w_{\kappa_5}(t)w_{\kappa_7}(t) \neq 0,\;\forall t \geq 0,
$$
we can solve this system to finally get the ASLO \eqref{www}.
\qed
%
%%%%%%%%%%%%%%%%%%%%
\section{Proof of Proposition \ref{pro4}}
\lab{appd}
%%%%%%%%%%%%%%%%%%%
%
\noindent {\bf ASLO for $\dot q_1$:} Given the identity
$$
\calf[\dot q_1]=\dot q_1-{1 \over \lambda}\calf[\ddot q_1],
$$
we obtain that
\begalis{
\dot q_1 &=\calf[\dot q_1]+{1 \over \lambda}\calf[\ddot q_1]\\
&=p\calf[q_1]+{1 \over \lambda}\calf[q_1 \dot q_2^2 +{u_1 \over m_1} ]\\
&=p\calf[q_1]+{1 \over \lambda}\calf\Big[q_1\Big({w_1 \over z_1 \calf[z_1]}\Big)^2+{u_1 \over m_1}\Big].
}
The second equation is obtained replacing $\dot q_2$ by its ALSO and the proof is completed replacing the second right hand term by $w_4$.\\

\noindent {\bf ASLO for $\dot q_2$:} First, we write from \eqref{rlsys} $\ddot q_2$ as
\begali{
 \ddot q_2 & =-{2 \over q_1} \dot q_1 \dot q_2 + {1 \over m_1 q_1^2}u_2, \; =-2\dot z_3 \dot q_2+z_2,
\lab{for1}
}
where we used the definition of $z_2$ in \eqref{zrl} and defined the measurable signal $z_3:=\ln(q_1).$

Define now the signal 
\begequ
\lab{chirl}
\chi:=z_1 \dot q_2,
\endequ 
whose derivative takes the form
\begali{
\dot \chi &= \dot z_1 \dot q_2 + z_1 \ddot q_2= \dot z_1 \dot q_2 + z_1(-2\dot z_3 \dot q_2+z_2) \nonumber \\
&= 2 z_1 \dot z_3\dot q_2-2z_1\dot z_3 \dot q_2+z_1 z_2= z_1z_2
\lab{dotchirl}
}
which is clearly measurable. Now, applying the filter $\calf$ to $\dot q_2={1 \over z_1}\chi$---defined in \eqref{chirl}---yields
\begalis{
p\calf[q_2]&=\calf\Big[\chi{1 \over z_1}\Big]\\
&=\chi\calf\Big[{1 \over z_1}\Big]-{1 \over \lambda}\calf \Big[\dot \chi \calf \Big[{1 \over z_1}\Big]\Big]\\
&=z_1 \dot q_2\calf\Big[{1 \over z_1}\Big]-{1 \over \lambda}\calf \Big[z_1z_2 \calf \Big[{1 \over z_1}\Big]\Big]
}
where we applied the SL to get the second identity and, for the third one, used the expressions for $\chi$ and $\dot \chi$ in \eqref{chirl} and \eqref{dotchirl}, respectively. The proof is completed rearranging terms to pull out the term $\dot q_2$. \\

\noindent {\bf ASLO for $\dot q_3$:}
%%%%%%%%%
%
Define the identity
$$
\ddot q_3=p[\dot q_3]+\lambda \dot q_3-\lambda \dot q_3,
$$
from which we get
$$
(p+\lambda)[\dot q_3]=p[\dot q_3]+\lambda \dot q_3.
$$
Dividing by $(p+\lambda)$ we get
\begalis{
\dot q_3 &={1 \over p+\lambda}[\ddot q_3]+\calf[\dot q_3] ={1 \over m_2}{1 \over p+\lambda}[u]+p\calf[q_3],
}
 completing the proof.
 %
%%%%%%%%%%%%%%%%%%%%
\section{Proof of Proposition \ref{pro5}}
\lab{appe}
%%%%%%%%%%%%%%%%%%%
%
\noindent {\bf ASLO for $\dot q_2$:} To simplify the derivation of the proof we first present the ASLO for $\dot q_2$. To streamline the proof we rewrite the equation for $\ddot q_2$ as follows
\begequ
\lab{for2}
\ddot q_2=-\dot z_3 \dot q_2+z_1,
\endequ
where we defined the measurable function $z_3:=\ln(\ell^2+q_1^2).$ Define now the function 
\begequ
\lab{Psi}
\Psi:=z_2 \dot q_2,
\endequ 
whose derivative satisfies
\begali{
\nonumber
\dot \Psi&=\dot z_2 \dot q_2+z_2 \ddot q_2=\dot z_2 \dot q_2+z_2(-\dot z_3 \dot q_2+z_1)\\
&=z_2 \dot z_3 \dot q_2+z_2(-\dot z_3 \dot q_2+z_1)=z_1 z_2
\lab{z1z2}
}
where, to obtain the second identity, we used the fact that $\dot z_2=z_2 \dot z_3$.

Now, from \eqref{Psi} we have that $\dot q_2={1 \over z_2} \Psi$, to which we apply the SL to get
\begalis{
\calf[\dot q_2]&=\calf[ \Psi{1 \over z_2}]=\Psi\calf\Big[{1 \over z_2}\Big]-{1 \over \lambda} \calf\Big[\dot \Psi \calf\Big[{1 \over z_2}\Big]\Big]=p\calf[q_2].
}
From the last identity we obtain
\begalis{
\Psi &={1 \over \calf\Big[{1 \over z_2}\Big]}\Big(p\calf[q_2]+{1 \over \lambda} \calf\Big[\dot \Psi \calf\Big[{1 \over z_2}\Big]\Big]  \Big)\\
 &={1 \over \calf\Big[{1 \over z_2}\Big]}\Big(p\calf[q_2]+{1 \over \lambda} \calf\Big[z_1z_2\calf\Big[{1 \over z_2}\Big]\Big]  \Big)=z_2 \dot q_2.
}
where we used \eqref{z1z2} to get the second equation. The proof is completed dividing the right hand side of the second equation above by $z_2$ and using the definitions given in Proposition \ref{pro5}.\\

\noindent {\bf ASLO for $\dot q_1$:} In the following derivations we assume that $\dot q_2$ is {\em known}, as it is given by the ASLO of Proposition \ref{pro5}. We find convenient to write the dynamics of $\ddot q_1$ in a state-space form. For, we define
\begalis{
\dot q_{11}&=q_{12},\;\dot q_{12}=q_{11}\dot q_2^2-g \sin(q_2),\;q_1=q_{11}.
} 
Proceeding from the identity
\begalis{
\calf[q_{12}]&=q_{12}-{1 \over \lambda}p\calf[q_{12}]=q_{12}-{1 \over \lambda}\calf[ q_{11}\dot q_2^2-g \sin(q_2)],
}
we get that
\begalis{
q_{12}
&=\calf[q_{12}]+{1 \over \lambda}\calf[ q_{11}\dot q_2^2-g \sin(q_2)]\\
&= p\calf[q_{1}]+{1 \over \lambda}\calf[ q_{1}\dot q_2^2-g \sin(q_2)].
}
The proof is completed noting that $q_{12}=\dot q_1$.


\begin{thebibliography}{00}
%
\bibitem{ARAetal_tac17}
S. Aranovskiy, A. Bobtsov, R. Ortega and A. Pyrkin, Performance enhancement of parameter estimators via dynamic regressor extension and mixing,  \TAC, vol. 62, pp. 3546-3550, 2017. 

\bibitem{BERbook}
P. Bernard, {\em Observer Design for Nonlinear Systems}, Springer, Switzerland, 2019.
	
\bibitem{BERANDAST}
P. Bernard, V. Andrieu, and D. Astolfi, Observer design for continuous-time dynamical systems, {\em Annual Reviews in Control}, pp. 224-248, 2022. 

\bibitem{ABAO2Automatica}
Bobtsov, A.A., Pyrkin, A.A., Ortega, R., Vukosavic, S.N., Stankovic, A.M., Panteley, E.V.: `A robust globally convergent position observer for the permanent magnet synchronous motor', \emph{Automatica},  2015, \textbf{61}, pp.~47-54.

\bibitem{data}
BMP Synchronous motor, Motor manual, V1.00, 12.2012, {\it Schneider Electric}, 2012.

\bibitem{BOBPYRORT}
A. Bobtsov, A. Pyrkin and R. Ortega. A new approach for  estimation of electrical parameters and flux observation of permanent magnet synchronous motors, {\it Int. J. on Adaptive Control and Signal Processing}, Vol, 30, No. 8-10, pp. 1434-1448, August-October, 2016.

\bibitem{BOBetal_ijc21} 
A. Bobtsov, R. Ortega, B. Yi and  N. Nikolayev,  Adaptive state estimation of state-affine systems with unknown time-varying parameters, \IJC, vol. 95, No. 9, pp.2460-2472,   2021.

\bibitem{BULLEWbook}
F. Bullo and A. Lewis, {\em Geometric Control of Mechanical Systems}, Springer Science-Bussiness Media, New York, 2005.

\bibitem{CHOetal}  
J. Choi, K. Nam, A. Bobtsov and R.  Ortega, Sensorless control of IPMSM based on regression model, \TPE, vol. 34, no. 9, pp. 9191-9201, 2019.

\bibitem{EFIPOL}
D. Efimov and A. Polyakov, Finite-time stability tools for control and estimation, {\em Annual Reviews in Control}, vol. 50, pp. 254-271, 2020.

\bibitem{HAUSASKOK}
J. Hauser, S. Sastry, and P. V. Kokotovi\'c, Nonlinear control via approximate input-output linearization: the ball and beam example, \TAC, IEEE vol. 37, no. 3, pp. 392-398, 1992.

\bibitem{ISIbook}
A. Isidori, {\em Lectures in Feedback Design for Multivariable Systems}, Springer, 2017.

\bibitem{KRAbook}
P. C. Krause, {\em Analysis of Electric Machinery}, New York: McGraw Hill, 1986.

\bibitem{KRSKOK}
M. Krsti\'c and P. V. Kokotovi\'c, Adaptive Nonlinear Design with Controller-Identifier Separation and Swapping, \TAC, vol 40, no. 3, pp. 426-440, 1995.

\bibitem{LEEetal} 
J. Lee, J. Hong, K. Nam, R. Ortega, A. Astolfi and L. Praly, Sensorless control of surface-mount permanent magnet synchronous motors based
on a nonlinear observer, {\em IEEE Trans. Power Electron.},  vol. 25, no. 2, pp. 290-297, 2010.

\bibitem{LIUetal}
X. Liu, G. Verghese, J. Lang and M. \"Onder, Generalizing the Blondel- Park Transformation of Electrical Machines: Necessary and Sufficient Conditions, \TCS, vol. 36, no. 8. pp. 1058-1067, 1989.

\bibitem{MARTOMVERbook}
R. Marino, P. Tomei and C. Verrelli, {\em Induction Motor Control Design} Berlin, Germany: Springer Science \& Business Media; 2010.

\bibitem{MASbook}
E. H. Maslen,  Self-sensing magnetic bearings, in G. Schweitzer and E.H. Maslen (Eds.), {\em Magnetic Bearings: Theory, Design, and Application to Rotating Machinery}, Springer, Berlin/Heidelberg, Chapter 15, pp. 435-459.

\bibitem{MOR}
A. S. Morse, Global stability of parameter-adaptive control systems, \TAC,  vol. 25, no. 3, pp. 433-439 1980.

\bibitem{NAMbook}
K.H. Nam, {\em AC Motor Control and Electric Vehicle Application}, CRC Press, 2010.

\bibitem{ORTetal_bookel} 
R. Ortega, A. Loria, P. J. Nicklasson and H. Sira--Ramirez, {\em Passivity--Based Control of  Euler--Lagrange Systems},  Springer-Verlag, Berlin, Communications and Control Engineering, 1998.

\bibitem{ORTetal}
R. Ortega, L. Praly, A. Astolfi, J. Lee and K. Nam, Estimation of rotor position and speed of permanent magnet synchronous motors with guaranteed stability, {\em IEEE Transaction on Control Systems Technology}, Vol. 19, No. 3, pp. 601-614, 2011

\bibitem{ORTetal_scl25} 
R. Ortega, B. Yi, A. Bobtsov,  A. Pyrkin and M. Sinetova, Adaptive state observers for a class of nonlinear systems: The non-triangular and non-state-affine case, \SCL, Vol. 196, 106020, 2025.

\bibitem{PYRetal}  
A. Pyrkin, A. Bobtsov, R. Ortega, A. Vedyakov and S. Aranovskiy, Adaptive state observer design using dynamic regressor extension and mixing, \SCL, Vol. 133, pp. 1-8, 2019.

\bibitem{RUGbook}
W.J. Rugh, {\em Linear Systems Theory}, 2nd Edition, Prentice hall, NJ, 1996.

\bibitem{SASBODbook}
S. Sastry and M. Bodson, {\em Adaptive Control: Stability, Convergence and Robustness}, Prentice-Hall, New Jersey, 1989.

\bibitem{VANbook}
A.~van~der Schaft, {\em $L_2$-Gain and Passivity Techniques in Nonlinear Control}, Springer, Berlin, 3rd Edition, 2016.

\bibitem{YIetal} 
B. Yi, R. Ortega, J. Choi and K. Nam, A high performance globally exponentially stable sensorless observer for the IPMSM: Theoretical and experimental results, \AUT, 174, {\tt (AUT-112138)}, 2025.

\end{thebibliography}
\end{document}